\documentclass[12pt]{article}
\usepackage{amsmath,amsfonts,amsthm,amssymb,amscd,colordvi}
\usepackage{graphicx}
\usepackage{enumerate}
\usepackage{latexsym}
\usepackage[usenames]{color}
\renewcommand{\Re}{\mathop{\rm Re}\nolimits}
\renewcommand{\Im}{\mathop{\rm Im}\nolimits}
\newcommand{\p}{\partial}

\newcommand{\e}{\varepsilon}
\newcommand{\vk}{\varkappa}

\newcommand{\vp}{\varphi}

\newcommand{\gi}{\rho}

\newcommand{\lla}{\gamma}

\newcommand{\ts}{\tilde s}

\newcommand{\bk}{{\mathbf k}}
\newcommand{\be}{\begin{equation}}
\newcommand{\ee}{\end{equation}}

\newcommand{\bb}{\mbox{\boldmath$\beta$}}
\newcommand{\R}{{\mathbb R}}
\newcommand{\C}{{\mathbb C}}
\newcommand{\IP}{{\bf P}}
\newcommand{\Z}{{\mathbb Z}}
\newcommand{\ZZ}{{\mathbb Z}^\infty_{+0}}
\newcommand{\E}{{\bf E}}

\newcommand{\T}{{\mathbb T}}

\newcommand{\N}{{\mathbb N}}
\newcommand{\PP}{{\bf P}}

\newcommand{\CC}{C_{pql}(2I)^pv^q\bar v^l}
\newcommand{\ai}{a}

\newcommand{\cA}{{\cal A}}
\newcommand{\cB}{{\cal B}}
\newcommand{\cD}{{\cal D}} 
\newcommand{\cF}{{\cal F}}

\newcommand{\cH}{{\cal H}}

\newcommand{\cHR}{{\cal H}^{\text{res}}}

\newcommand{\cM}{{\cal M}}
\newcommand{\cQ}{{\cal Q}}

\newcommand{\cR}{{\cal R}}

\newcommand{\cT}{{\cal T}}

\newcommand{\strela}{\rightharpoonup}
\newcommand{\lc}{\lceil}
\newcommand{\rc}{\rceil}

\newcommand{\const}{\mathop{\rm const}\nolimits}
\newcommand{\as}{\quad\mathop{\rm as} \nolimits\quad}

\newcommand{\supp}{\mathop{\rm supp}\nolimits}

\newcommand{\Arg}{\mathop{\rm Arg}\nolimits}

\def\dbar{{\mathchar'26\mkern-12mu d}}
\def\12{\tfrac12}

\def\lan{\langle}
\def\ran{\rangle}

\def\eps{\varepsilon}
\theoremstyle{plain}
\newtheorem{theorem}{Theorem}[section]
\newtheorem{lemma}[theorem]{Lemma}
\newtheorem{proposition}[theorem]{Proposition}
\newtheorem{corollary}[theorem]{Corollary}
\theoremstyle{definition}

\theoremstyle{remark}

\numberwithin{equation}{section}

\begin{document}

\author{Sergei
Kuksin\footnote{CNRS and  I.M.J, Universit\'e Paris Diderot-Paris 7, Paris, 
 France, e-mail:
  kuksin@math.jussieu.fr },\addtocounter{footnote}{2} Alberto Maiocchi \footnote{Laboratoire de
Math\'ematiques,  Universit\'e de Cergy-Pontoise, 2  avenue Adolphe Chauvin,
Cergy-Pontoise, France,
e-mail: alberto.maiocchi@unimi.it}}

\title{ Resonant averaging for small solutions of  stochastic NLS equations }
\date{}
\maketitle

\begin{abstract}
We consider the free linear Schr\"odinger  equation  on a torus
$\T^d$, perturbed by a hamiltonian  nonlinearity,
driven by a random force and damped by 
a linear damping:
$$
 u_t -i\Delta u +i\nu \rho |u|^{2q_*}u
= - \nu f(-\Delta) u
+ \sqrt\nu\,\frac{d}{d t}\sum_{\bk\in \Z^d} b_\bk\bb^\bk(t)e^{i\bk\cdot x} \ .
$$
Here $u=u(t,x),\ x\in\T^d$, 
 $0<\nu\ll1$,  $q_*\in\N$, 
 $f$ is a positive continuous function, $\rho$ is a positive parameter 
  and 
$\bb^\bk(t)$ are standard independent complex Wiener processes.  We are interested in limiting, as
$\nu\to0$, behaviour of distributions of solutions for this equation 
and of its stationary measure. Writing the equation in the slow time $\tau=\nu t$, we
 prove that the limiting  behaviour of the  both 
  is  described by the   {\it effective equation}
$$
u_\tau+ f(-\Delta) u = -iF(u)+\frac{d}{d\tau}\sum  b_\bk\bb^\bk(\tau)e^{i\bk\cdot x} \ ,
$$
where the nonlinearity $F(u)$ is made out of the resonant terms of the monomial $ |u|^{2q_*}u$. 
We explain the relevance of this result for the problem of weak turbulence 
\end{abstract}

\tableofcontents

\section{Introduction}\label{s0}
\subsection
{ Equations}  \label{s0.1}
We study the free Schr\"odinger equation on the torus $\T^d_L= \R^d/(2\pi L\Z^d)$,
\begin{equation}\label{*1}
u_t(t,x)-i\Delta u(t,x)=0,\quad x\in\T^d_L\,, 
\end{equation}
stirred by a perturbation, which comprises a
hamiltonian  term, a linear damping and a random force. That is, we
consider the equation  
\begin{equation}\label{1.11}
\begin{split}
u_t-i\Delta u= - i\eps^{2q_*}|u|^{2q_*}u-\nu f(-\Delta) u
+\sqrt\nu \frac{d}{dt}\sum_{\bk\in\Z^d_L}  b_\bk\bb^\bk(t) e^{i{\bk}\cdot x}
\ ,\\
 u=u(t,x),\quad x\in \T^d_L\,,
\end{split}
\end{equation}
where   $\ q_*\in \N $ and
$  \eps,\nu>0$ are two small parameters, controlling  the size of the
perturbation, while $\Z^d_L$ denotes the set of vectors of the form
$\bk={\mathbf{l}}/L$ with ${\mathbf{l}}\in \Z^d$. 
 The damping  $-f(-\Delta) $ is the selfadjoint 
  linear operator in $L_2(\T^d_L)$ which acts on the exponents $e^{i{\bk}\cdot x}$, $\bk\in\Z^d_L$, 
 according to
\begin{equation}\label{f}
f(-\Delta) e^{i{\bk}\cdot x}= \lla_\bk e^{i{\bk}\cdot x}, \qquad \lla_\bk= f( \lambda_\bk  )
\quad\text{where}\quad \lambda_\bk= |\bk|^2\ .
\end{equation}
 The real-valued smooth function $f(t), t\ge0$,    is  positive  and $f'>0$. 
 To avoid technicalities, not relevant for  this work, we assume that 
$\ 
f(t)\ge C_1 |t| +C_2$ for all $t$,
for suitable positive constants $C_1,C_2$ (for example, $f(-\Delta)u=-\Delta u+u$). 
The processes $\bb^\bk, \bk\in\Z^d_L$, are
standard independent complex 
Wiener processes, i.e., $\bb^\bk(t)=\beta^\bk_+(t)+i\beta^{\bk}_-(t)$, 
where $\beta_{\pm}^\bk(t)$ are standard independent  real Wiener
processes. The real numbers $b_\bk$ are all non-zero and decay  fast
when $|\bk|\to\infty$. 

The nonlinearity  in \eqref{1.11} is hamiltonian and may be written as
\begin{equation}\label{*ham}
-i\e^{2q_*} |u|^{2q_*}u=\e^{2q_*}i   \,\nabla \cH(u) ,\quad
\cH(u) = \cH^{2q_*+2}(u)
=-\frac1{2q_*+2}\int |u(x)|^{2q_*+2}dx.
\end{equation}
We {\it assume} that eq. \eqref{1.11} with sufficiently smooth initial data $u_0(x)$
is well posed. It is well known that this assumption holds (at least) under
some restriction on $d,q_*$ and the growth of $f(t)$ at infinity, see in  Section~\ref{s1.1}

Equation \eqref{1.11} with small $\nu$ and $\eps$  is important for
physics and mathematical physics, where it serves  as a universal
model. In particular, it is  used in the physics of plasma 
to describe small oscillations of the media on long  time scale, see 
 \cite{Fal, Naz, ZL75, ZLF}. 
 The parameters $\nu$ and $\eps$ measure, respectively, the inverse 
  time-scale of the forced oscillations, and their  amplitude. 
  Physicists consider different regimes, where the two parameters
  are tied in various ways. To do this they assume some relations between $\e$ and $\nu$,
  explicitly or  implicitly. In our work we choose
  $$
  \eps^{2q_*}=\gi \nu,
  $$
  where $\rho>0$ is a constant. This assumption is within the  usually imposed bounds,
  see  \cite{ Naz}.   Passing to the slow time
$\tau=\nu t$, we  get the rescaled equation 
\begin{equation}\label{1.111}
\begin{split}
\dot u+i\nu^{-1}\big(-\Delta u\big)= -f(-\Delta) u- i\gi|u|^{2q_*}u
+
\sum_{\bk\in\Z^d_L}  b_\bk \dot\bb^\bk(\tau) e^{i{\bk}\cdot x}
\ ,
\end{split}
\end{equation}
where $u=u(\tau,x)$, $x\in \T^d_L$ and 
 the upper dot   $\ \dot{}\ $  stands for 
$\frac{d}{d\tau}$.  If we write $u(\tau,x)$ as Fourier series,
$\ 
u(\tau,x)=\sum_\bk v_\bk(\tau)e^{i\bk\cdot x},
$
then in view of \eqref{*ham},  eq. \eqref{1.111} may be written as the system
\begin{equation}\label{1.100}
\begin{split}
\dot v_\bk+i \nu^{-1}\lambda_\bk v_\bk=-\gamma_\bk v_\bk + 2\rho\, i   \,\frac{\p \cH(v)}{\p \bar v_\bk}  
+b_\bk \dot\beta^\bk(\tau),\quad \bk \in \Z^d_L.
\end{split}
\end{equation}
Here $\cH(v)$ is the Hamiltonian $\cH$, expressed in terms  of the Fourier coefficients
$v=(v_\bk, \bk\in\Z^d_L)$:
\begin{equation}\label{Hv}
\cH(v)=-\frac1{2q_*+2}
\sum_{\bk_1,\dots \bk_{2q_*+2}\in\Z^d_L} v_{\bk_1}\dots v_{\bk_{q_*+1}} \bar v_{\bk_{q_*+2}} \dots
\bar v_{\bk_{2q_*+2}} \,  \delta^{1\ldots q_*+1}_{q_*+2\ldots 2q_*+2}\,,
\end{equation}
and we use a notation, standard in physics (see \cite{Naz}):
\begin{equation}\label{N1}
\delta^{1\ldots q_*+1}_{q_*+2\ldots 2q_*+2}=\left\{\begin{array}{cc}
1 & \mbox{if }\bk_1+\ldots+\bk_{q_*+1}-\bk_{q_*+2}-\ldots-\bk_{2q_*+2} = 0 \\
0 & \mbox{otherwise}
\end{array}
\right.\ .
\end{equation}
As before  we
are interested in  the limit $\nu\to0$, corresponding   to small oscillations in the original
non-scaled equation.\footnote{See \cite{KNer13} for a theory of  equation \eqref{1.111} 
 for the case
when $f(t)=t+1$ and $\nu=\infty$.}

We note that the method of our work applies as well to equations \eqref{1.100} with  the Hamiltonians
$\cH$ of the form \eqref{*ham}, where the density of the Hamiltonian is a real-valued polynomial 
of $u$ and $\bar u$ (not necessarily a polynomial of $|u|^2$). For instance, we could work with
the cubic Hamiltonians $\cH^3=\int |u|^2(u+\bar u)\,dx$ or $\cH^3=\int (u^3+\bar u^3)\,dx$.

\subsection
{Weak Turbulence} \label{sWT}
In physics equations \eqref{1.111} with $\nu\to0$ are treated by the theory of weak turbulence,
or 
 WT (this abbreviation also may stand for `Wave Turbulence', but the difference between  the two notions
seems for us negligible);  see the  works, quoted above as well as \cite{CZ00}.
 That theory either deals with  equation \eqref{1.111},
where  $L=\infty$ 
by formal replacing Fourier series for $L$-periodic functions with Fourier integrals and makes with them bold 
transformations, or  considers the limit  $\nu\to0$ simultaneously with the limit  $L\to\infty$ and 
treats the two of  them in an equally bold way.\footnote{
Alternatively (and more often) people, working on WT, consider the HPDE \eqref{1.111}${}_{f=0, b_\bk=0\,\forall \bk}$
and treat it in a similar formal way, see \cite{Fal, Naz, ZLF, CZ00}. 
The corresponding problems do not fit our technique. Some recent progress in 
their rigorous study may be found in \cite{FGH}.
}
Concerning this limit WT makes a number of remarkable 
predictions, based on tools and ideas, developed in the community, which can be traced 
 back to the work \cite{Peierls}. The most famous of them deals with the
energy spectrum of solutions $u(\tau,x)$.  To describe the corresponding claims,
 consider the quantity $\E|v_\bk(\tau)|^2$, average it in
 time\footnote{Certainly this is not needed if we consider stationary solutions of the equation.}
 $\tau$  
 and in wave-vectors 
$\bk\in \Z^d_L$  such that $|\bk| \approx r>0$; next properly scale this and denote the result $E_r$.
The function $r\to E_r$ is called the {\it energy spectrum}.  It is predicted by WT that, in certain
{\it inertial range} $[r_1,r_2]$, which is contained in the spectral zone 
where the random force is negligible (i.e., where $|b_\bk|\lll (\E|v_\bk|^2)^{1/2}$ if $r_1\le |\bk| \le r_2$), the
energy spectrum has an algebraic behaviour:
\begin{equation}\label{KZ}
E_r\sim r^{-\alpha}\quad \text{for}\quad  r\in[r_1,r_2],
\end{equation}
for a suitable $\alpha>0$.  The WT limit, in fact, deals with the double limit:
\begin{equation}\label{limits}
\begin{split}
L\to\infty,\quad \nu\to0  \, .\quad
\end{split}
\end{equation}
 Relation  between the two parameters in \eqref{limits}
is not quite clear, and it may be better to talk about the WT limits (rather then about a single case). Note that
 only the limits which  lead to  relations \eqref{KZ} with  finite $\alpha$'s are relevant for the WT.

We suggest to study  the WT limits (at least, some of them) by splitting the limiting process 
  in two steps:

I) prove first that when $\nu\to0$, statistical characteristics of solutions $u^\nu$ have limits of order one, 
described by certain {\it effective equation} which is a nonlinear stochastic equation with coefficients 
of order one and with a hamiltonian  nonlinearity, made out the resonant terms of the 
nonlinearity $ |u|^{2q_*}u$.

II) Show then that main statistical  characteristics of solutions for the effective equation have non-trivial limits of 
order one, when $L \to\infty$ and $\rho=\rho(L)$ is a suitable function of $L$.

In this work we perform Step I, postponing  Step II for the future. We
stress  that the results of Step I 
alone cannot justify the predictions of WT since the latter  (e.g. the
asymptotic \eqref{KZ}) cannot hold  
when the period $L$ is fixed and finite. On the other hand, as we show
in a heuristic way in \cite{KM13KZ}, a suitable choice of the function
$\rho(L)$ leads  to a Kolmogorov-Zakharov type kinetic equation in the
limit of $L\to \infty$ and, as a consequence, to energy spectra of the desired form \eqref{KZ}.
This encourages us  to pursuit  our program, which
 brings to WT the advantage of a rigorous foundation, based on  the recent
results of stochastic calculus.  
 It is open for discussion up to what extent the corresponding  choice of the limits in \eqref{limits} 
 agrees with physics and the tradition of WT. Still we believe that no matter what the result
 of this discussion is, the Step~I), performed in this work,  Step~II), whose rigorous realisation 
 is left for future, and their synthesis   are interesting and important by themselves. 

 As the title of the paper suggests, our argument is a form of averaging.
The latter is a tool which is used by the WT community on a regular basis, either explicitly (e.g. see \cite{Naz}), or
implicitly.

\subsection
{  Inviscid limits for damped/driven hamiltonian PDE,  effective equations and interaction representation
}  Equation \eqref{1.11} is the linear hamiltonian PDE (HPDE) 
\eqref{*1}, driven by the random force, damped by the linear damping $-\nu f(-\Delta u)$ and perturbed 
by the hamiltonian nonlinearity $-\e^{2q_*} i\rho|u|^{2q_*}u$.   Damped/driven
 HPDE and the inviscid limits in these equations when the random force and the 
damping go to zero, are very important for physics. In particular, since the $d$-dimensional Navier-Stokes
equation (NSE) with a random force can be regarded as a damped/driven Euler equation (which is an HPDE), and
the inviscid limit for  the NSE describes the $d$-dimensional turbulence. The NSE with random force, 
 especially when
$d=2$, was intensively studied last years, but the corresponding inviscid limit turned out to be very 
complicated even for $d=2$, see \cite{KS}. The problem of this  limit 
becomes feasible when the underlying HPDE is integrable or linear. The most famous integrable PDE is the KdV
equation. Its damped/driven perturbations and the corresponding inviscid limits were studied   in \cite{KP08, K10}.  In \cite{K12} the method of those works was applied to 
the situation when the unperturbed HPDE is the Schr\"odinger equation
\begin{equation}\label{*5}
u_t + i(-\Delta u + V(x) u)=0,\qquad x\in\T^d_L\,,
\end{equation}
where the potential $V(x)$ is in general position. Crucial for the just mentioned works  is that there the 
unperturbed equation is free from strong resonances. For \cite{KP08, K10} it means that all solutions of KdV
are  almost-periodic functions of time, and for a typical solution the corresponding frequency vector is 
free from resonances; while for  \cite{K12}  it means that for the typical potentials $V(x)$, considered in 
\cite{K12}, the spectrum of the linear operator in \eqref{*5} is non-resonant. 

In contrast, now the linear operator in the unperturbed equation \eqref{*1}  has the eigenvalues 
$\lambda_\bk, \bk\in\Z^d_L$ (see \eqref{f}), which are highly
resonant (accordingly, all solutions for eq. \eqref{*1} are periodic with the same period $2\pi L^{-2}$).  
This gives rise to an additional difficulty.
To explain it, we  rewrite equation
\eqref{1.111}=\eqref{1.100} as a fast-slow system, denoting 
$
I_\bk=\tfrac12 |v_\bk|^2, \;\; \vp_\bk=\Arg v_\bk
$
(these are the action-angles for the linear hamiltonian system \eqref{*1}). 
In the new variables  eq.~\eqref{1.111} reads 
\begin{equation}\label{*6}
\dot I_\bk(\tau) = v_\bk\cdot P_\bk(v) +b_\bk^2 +b_\bk(v_\bk\cdot\dot\bb^\bk), 
\end{equation}
\begin{equation}\label{*7}
\dot \vp_\bk(\tau) =-\nu^{-1} \lambda_\bk +I_\bk^{-1}\cdot\dots, 
\end{equation}
where $\bk\in\Z^d_L$ and the dot $\cdot$ indicates the real scalar product in $\C\simeq \R^2$.
Here $P(v)$ is the vector field in the r.h.s. of the  $v$-equation \eqref{1.100} and $\dots$ abbreviates  a term of order one (as $\nu\to0$). 
  If the frequencies $\{\lambda_\bk\}$ are  resonant, then equations for some linear 
combinations of the phases $\vp_\bk$ are slow, which make it more difficult to analise the system. The
method of resonant averaging treats this problem in finite dimension, see \cite{AKN} and
Section~\ref{s3.1} below. In the situation at hand, we have additional problem: the $\vp$-equations
\eqref{*7}  have singularities at the  locus 
\begin{equation}\label{game}
\Game=\{ I: I_\bk =0 \;\;\text{for some} \; \bk\}
\end{equation}
which is dense in the space of sequences $(I_\bk, \bk\in\Z^d_L)$, 
and the averaged $I$-equations 
\begin{equation}\label{*8}
\dot I_\bk(\tau) =\lan v_\bk\cdot P_\bk\ran(I) +b_\bk^2 +b_\bk \sqrt{2I_\bk}\,\dot\beta^\bk(\tau)\,,
\quad \bk\in  \Z^d_L\,,
\end{equation}
where $\lan\cdot\ran$ signifies the average in $\vp\in\T^\infty$, 
have there weak singularities.  A way to overcome these difficulties 
is to find for \eqref{*6}, \eqref{*7} an {\it effective equation}, which is a system of regular 
equations 
\begin{equation}\label{*eff}
\dot v_\bk = R_\bk(v) + b_\bk\dot\bb^\bk(\tau), \quad \bk\in\Z^d_L, 
\end{equation}
such that under the natural projection
$\ 
v_\bk \mapsto I_\bk = \tfrac12 |v_\bk|^2$, $\bk\in\Z^d_L,
$
 solutions of \eqref{*eff} transform to solutions of \eqref{*8}.
In   \cite{K10} this approach was used to study the perturbed KdV equation,
written as a fast-slow  system, similar to \eqref{*6}, \eqref{*7}. 
 That system has strongly non-linear behaviour,
and in \cite{K10} the effective equation was constructed as a kind of averaging 
of the $I$-equations. In \cite{K12} an effective equation for the damped/driven 
nonresonant equation \eqref{*5} was derived in a similar way. If the introduced 
damping is linear and the nonlinearity is hamiltonian, like in eq.~\eqref{1.11}, then the 
effective equation in \cite{K12} is  linear. 

When the unperturbed hamiltonian system is linear,  an alternative way to find an 
effective equation is to use the {\it interaction
  representation}. I.e., to pass from the  
complex variables $v_\bk(\tau)$ (which diagonalise the linear system) to the fast 
rotating variables
\begin{equation}\label{a}
a_\bk (\tau) =e^{i\nu^{-1}\lambda_\bk \tau} v_\bk(\tau),\qquad \bk\in\Z^d_L. 
\end{equation}
Since $|a_\bk|=|v_\bk|$, then the limiting dynamics of the $a$-variables controls
the limiting behaviour of the actions $I_\bk$. So a regular system of equations, 
describing the limiting $a$-dynamics, is the effective equation. 
N.~N.~Bogolyubov used this approach for the finite-dimensional 
deterministic averaging, calling it {\it averaging in the quasilinear systems} (see in \cite{AKN}). 
The interaction
representation is systematically  used in the WT.

\subsection
{  Resonant Hamiltonian $\cHR$}
Now consider the fast-slow equations \eqref{*6}, \eqref{*7} which come from 
eq.~\eqref{1.100}, and where the fast motion \eqref{*7}  is highly resonant. Repeating
the construction of the effective equation from \cite{K12}, but replacing there the usual 
averaging by the resonant averaging, we find 
 an effective equation, corresponding to
\eqref{1.100}. It turned  out to be another damped/driven hamiltonian system
with a Hamiltonian $\cHR$, obtained by the resonant averaging of $\cH(v)$,
see Section~\ref{s2.3}.  As we said above, an
 alternative  way to derive the effective equation is through 
the interaction representation, i.e., by transition from the $v$-variables to the $a$-variables \eqref{a}.
In view of \eqref{1.100}, the $a$-variables satisfy the system of equations
\begin{equation}\label{1.100a}
\begin{split}
\dot a_\bk=&-\gamma_\bk a_\bk  + e^{i\nu^{-1} \lambda_{\bk} \tau} b_\bk
\dot\beta^\bk(\tau) \\
&-\rho\, i  \sum_{\bk_1,\dots
  \bk_{2q_*+1}\in\Z^d_L} a_{\bk_1} \dots a_{\bk_{q_*+1}}\bar
  a_{\bk_{q_*+2}} \dots \bar a_{\bk_{2q_*+1}} \delta^{1\ldots
    q_*+1}_{q_*+2\ldots 2q_*+1\, k}\\
& \times \exp\left(-i\nu^{-1}\tau(
  \lambda_{\bk_1}+\dots+\lambda_{\bk_{q_*+1}}-
  \lambda_{\bk_{q_*+2}}-\dots
  -\lambda_{\bk_{2q_*+1}}-\lambda_{\bk})\right),\;\; \bk\in \Z^d_L. 
\end{split}
\end{equation}
The  terms,   constituting the nonlinearity,  oscillate 
fast  as $\nu$ goes to zero, unless the sum of the eigenvalues in
the exponent in the 
 third line vanishes. This leads  to the right  guess that only the terms
for which this sum  equals  zero (i.e., the resonant terms), contribute
to the limiting dynamics, and that the effective equation is the
following damped/driven hamiltonian system
\begin{equation}\label{*eff1}
\begin{split}
\dot v_\bk=-\gamma_\bk v_\bk + 2\rho\, i   \,\frac{\p \cH^{\text{res}}(v)}{\p \bar v_\bk}  
+b_\bk \dot\beta^\bk(\tau),\quad \bk\in\Z^d_L\,.
\end{split}
\end{equation}
Here the Hamiltonian $\cHR(v)$ is given by the sum 
\begin{equation}\label{hres}
- \frac1{2q_*+2}
\sum_{\bk_1,\dots \bk_{2q_*+2}\in\Z^d_L} v_{\bk_1}\dots v_{\bk_{q_*+1}} \bar v_{\bk_{q_*+2}} \dots
\bar v_{\bk_{2q_*+2}} \,  \delta^{1\ldots q_*+1}_{q_*+2\ldots 2q_*+2}\,
 \delta(\lambda^{1\ldots q_*+1}_{q_*+2\ldots 2q_*+2})\,,
\end{equation}
so that $2\rho i \frac{\p \cH^{\text{res}}}{\p \bar v_\bk}  (v)$ is
\begin{equation}\label{vfield}
 -\rho i
  \sum_{\bk_1,\dots \bk_{2q_*+2}\in\Z^d_L} v_{\bk_1}\dots v_{\bk_{q_*+1}} \bar v_{\bk_{q_*+2}} \dots
\bar v_{\bk_{2q_*+1}} \,  \delta^{1\ldots q_*+1}_{q_*+2\ldots 2q_*+1\, k}\,
 \delta(\lambda^{1\ldots q_*+1}_{q_*+2\ldots 2q_*+1\, k})\,,
\end{equation}
where we use another physical abbreviation:
\begin{equation}\label{N2}
\delta(\lambda^{1\ldots q_*+1}_{q_*+2\ldots 2q_*+2})=\left\{\begin{array}{cc}
1 & \mbox{if }\lambda_{\bk_1} +\ldots+\lambda_{\bk_{q_*+1}}
-\lambda_{\bk_{q_*+2}}- \ldots- \lambda_{\bk_{2q_*+2}}
 = 0, \\
0 & \mbox{otherwise.}
\end{array}
\right.\ .
\end{equation}
This representation for $\cHR$ is different from that given by the resonant averaging. 
Its advantage is the natural relation with the  $a$-variables, which is
convenient to study the limit $\nu\to0$. The 
 representation for $\cHR$ by means of the resonant averaging  turned out to be 
more useful  to study properties of $\cHR$ and of  the corresponding
hamiltonian vector field. 

We saw that the effective equation can be obtained from the system 
 \eqref{1.100} by a simple procedure:  drop the fast rotations
and replace the Hamiltonian $\cH$ by its resonant average $\cHR$. 
In difference with the non-resonant case, this is a nonlinear system. 
The corresponding hamiltonian equation
\begin{equation}\label{hamm}
\dot v_\bk= 2\rho\, i   \,\frac{\p \cH^{\text{res}}(v)}{\p \bar v_\bk} \,,
 \quad \bk\in\Z^d_L\,,
\end{equation}
has a vector field, 
 locally Lipschitz in sufficiently smooth spaces, so the equation is well
posed locally in time. In fact, it is globally well posed. 
 We get  this result  in Section~\ref{s5.2} as a simple 
 consequence of our  main theorems.
 
 The Hamiltonian  $\cHR$ has two  convex quadratic integrals,
 $$
 H_0(v)=\frac12\sum |v_\bk|^2,\qquad  H_1= \frac12 \sum\lambda_\bk |v_\bk|^2\,,
 $$
 which are 
 similar to the energy and the enstrophy integrals for  the 
  2d~Euler equation on $\T^2$ (see  \eqref{integrals}),
  and the vector-integral of moments
  $$
  M(u) = \frac12 \sum \bk |u_\bk|^2\in \R^d\,,
  $$
  which can be compared with the extra integrals of the 2d~Euler. Besides, the 
   vector-field \eqref{vfield}  is non-linear homogeneous and hamiltonian, 
   as  that of the Euler equation.
  This makes  the effective equation \eqref{*eff1}    similar 
   to the 2d~Navier-Stokes system on $\T^2$. Fortunately for Step~II above, 
    the  former is significantly simpler then the latter.
\smallskip

 The construction of the resonant Hamiltonian $\cHR$ is in the spirit of WT, and the corresponding 
 hamiltonian equation \eqref{hamm} is known there as the {\it equation of discrete turbulence},
 see \cite{Naz},~Chapter~12.  Similar equations were  considered by  mathematicians, interested in related problems
 (see \cite{GG12}), and were used by them for intermediate arguments
  (e.g., see \cite{FGH}).
The stochastic equation \eqref{*eff1} was not considered before our work.

\subsection
{ Results}  Main results of our work  are stated and proved in Section~\ref{s.results},  based on properties of 
 the effective equation, established earlier.
 They imply  that the long-time behaviour of solutions for equations 
\eqref{1.111}, when $\nu\to0$, is controlled in distribution  by solutions for the effective equation. 
We start with the results on the Cauchy problem.  So, let $v^\nu(\tau)$ be a solution of 
\eqref{1.100} such that 
$$
v^\nu(0)=v_0,
$$
where $v_0=(v_{0\bk}, \bk\in\Z^d_L)$ corresponds to a sufficiently smooth function $u_0(x)$. Let us fix any $T>0$.

Consider the list  $\cA$ of resonances in eq.~\eqref{*1}.  That is,  the set of all nonzero integer 
vectors $\xi=(\xi_\bk, \bk\in\Z^d_L)$ of finite length, satisfying  $\sum_{\bk \in\Z^d_L} \xi_{\bk}\lambda_{\bk}=0$.  For
$\xi\in\cA$ consider the corresponding resonant combination of phases of solutions $v^\nu(\tau)$,
$$
\Phi^\xi(v^\nu(\tau)):=  \sum_{\bk \in\Z^d_L} \xi_{\bk} \vp_{\bk}(v^\nu(\tau))\in S^1=\R/2\pi\Z,\quad 0\le\tau\le T.
$$
Consider also the vector of actions 
$
I(v^\nu(\tau))=\{ I_\bk(v^\nu(\tau)), \bk\in \Z^d_L\}.
$

\noindent
{\bf Theorem 1}. When $\nu\to0$, we have the weak convergence of  measures 
$$
\cD\big(I(v^\nu(\tau)) \big) \strela  \cD\big(I(v^0(\tau))\big),
$$
where $v^0(\tau),\ 0\le\tau\le T$, is a unique solution of equation 
\eqref{*eff1} such that $v^0(0)=v_0$. 
\medskip

The distributions of 
resonant combinations of phases $\Phi^\xi(v^\nu(\tau))$, mollified in $\tau$, 
converge  to the mollified distributions to $\Phi^\xi(v^0(\tau))$, see Section~\ref{s5.2} for an exact 
statement of the result.  
 On the contrary, if a finite vector
$s=(s_\bk, \bk\in\Z^d_L)$ is non-resonant, i.e.,  $\sum s_\bk \lambda_\bk\ne0$, then the measure
$
\cD(\Phi^{(s)}(v^\nu(\tau)),
$
mollified in $\tau$, converges when $\nu\to0$ to the Lebesgue measure on $S^1$.  The theorem is
proved in Section~\ref{s5.2}, using the interaction representation \eqref{1.100a} for equation \eqref{1.11}.

The limiting behaviour of solutions $v^\nu(\tau)$ can be described without evoking the effective equation.
Namely, denote by $\cA_m$ the set of resonances $\xi\in\cA$ of length $|\xi|\le m:=2q_*+2$. Then 
the vectors $I^\nu(\tau)=I(v^\nu(\tau))$ and $\Phi^\nu(\tau)=\big(\Phi^\xi(v^\nu(\tau)), \xi \in\cA_m\big)$
converge in distribution to  limiting processes $I^0(\tau)$ and $\Phi^0(\tau)$, which are weak solutions 
of the corresponding averaged equations on these vectors. The equations
for $\Phi$ have strong singularities at the locus $\Game$, and rigorous formulation of this convergence is
involved, see Proposition~\ref{p.slow}.

\medskip

Now consider a stationary measure $\mu^\nu$ for  equation \eqref{1.111}  (it always exist). We have 

\noindent
{\bf Theorem 2}. Every sequence $\nu'_j\to0$ has a subsequence $\nu_j\to0$ such that 
$$
I\circ \mu^{\nu_j}\strela  I \circ m^0,\qquad
\Phi^{(\xi)}\circ \mu^{\nu_j}\strela \Phi^{(\xi)}\circ m^0\quad \forall\, \xi\in\cA,
$$
where $m^0$ is a stationary measure for equation \eqref{*eff1}. If a vector $s$ is non-resonant, then the
measure $\Phi^{(s)}\circ\mu^\nu$ converges, as $\nu\to0$, to the Lebesgue measure on $S^1$.
\medskip

If the effective equation has a unique stationary measure $m^0$, then the limits in Theorem~2 
do not depend on the sequence $\nu_j\to0$, so the convergences hold as $\nu\to0$. 
Remarkably, in this case 
the measure $m^0$ controls not only the fast, but also the slow components of the measures
$\mu^\nu$:

\bigskip
\noindent
{\bf Theorem 3}.  If the effective 
equation has a unique stationary measure $m^0$, then 
$
\mu^\nu \strela m^0$ as $ \nu\to0.
$
\medskip

In particular, if  the effective 
equation has a unique stationary measure $m^0$ and the equation \eqref{1.11} is mixing,\footnote{both these conditions
hold, e.g. if $q_*=1$ and $f(\lambda)=c_1+\lambda^{c_d}$, where $c_d$ is sufficiently big in terms of $d$.} 
then $m^0$ describes asymptotical behaviour of distributions of solutions $u(t)$ for \eqref{1.11} as $t\to \infty$
and $\nu\to0$:
$$
\lim_{\nu\to0}\lim_{t\to\infty}
\cD(u(t))=m^0.
$$

In view of the last theorem, it is important to understand when the effective equation has a unique
stationary measure and  is mixing.  This is discussed in Section~\ref{s.mix}. In particular, the 
mixing holds if $q_*=1$, $f(t)=t+1$ and $d\le3$. 
\smallskip

This work is a revised version of the preprint \cite{KM13}.

\medskip

\noindent
{\bf Notation and Agreement.} The  {\it stochastic terminology} we use agrees with \cite{KaSh}.
 All filtered probability spaces we work with satisfy the {\it usual condition} (see \cite{KaSh}).\footnote{
I.e., the corresponding filtrations $\{\cF_t\}$ are continuous from the right, and each $\cF_t$ contains all negligible sets.}
 Sometime we forget to mention that 
a certain relation holds a.s. \\
{\it Spaces of integer vectors.} We denote by $\Z^\infty_0$ the set of vectors in $\Z^\infty$ 
of finite length,  and denote $\ZZ=\{s\in\Z^\infty_0: s_k\ge0\  \forall\,k\}$. Also
see \eqref{notation} and \eqref{hren}. 

\noindent 
{\it Infinite vectors.} For an infinite vector $\xi=(\xi_1,\xi_2,\dots)$ (integer, real or complex) and $N\in\N$ 
we denote by $\xi^N$ the vector $(\xi_1,\dots,\xi_N)$, or the vector $(\xi_1,\dots,\xi_N,0,\dots)$, depending on the 
context. For a complex
vector $\xi$ and $s\in\ZZ$ we denote  $\xi^s=\prod_j \xi_j^{s_j}$. 

\noindent
  {\it Norms.} We use $|\cdot|$ to denote the Euclidean norm in $\R^d$ and in 
$\C\simeq\R^2$, as well as the $\ell_1$-norm in $\Z^\infty_0$. For the norms
$|\cdot|_{h^m}$ and $|\cdot|_{h^m_I}$
 see \eqref{vnorm} and below that.

\noindent
{\it Scalar products.}  The notation  ``$\cdot$'' stands for  the
scalar product in $\Z^\infty_0$, the paring of $\Z^\infty_0$ 
with $\Z^\infty$, the Euclidean scalar product in $\R^d$ and in $\C$. The latter means that if $u,v\in\C$, then 
$u\cdot v=\Re(\bar u v)$.  The $L_2$-product is denoted $\lan \cdot, \cdot\ran$, and we also denote by
$\lan f,\mu\ran = \lan\mu,f\ran$ the integral of a function $f$ against a measure $\mu$.

\noindent
{\it Max/Min.} We denote $a\vee b=\max(a,b)$, $a\wedge b=\min(a,b)$. 
\bigskip          

\medskip\par
\noindent{\it Acknowledgments.} We wish to thank for discussions and advice 
Sergey Nazarenko, Anatoli Neishtadt 
 and Vladimir Zeitlin. 
This work was supported by  l'Agence Nationale de la Recherche 
 through the grant STOSYMAP (ANR 2011BS0101501).

\section{Preliminaries}\label{s1}
Since in this work 
we are  not interested in the dependence  of the
results on  $L$, 
 from now on it will be kept fixed and equal to 1, apart from Section~\ref{s2.5}. There 
 we make explicit calculations, controlling how their results depend 
  on $L$.
\subsection{Apriori  estimates.}\label{s1.1}
In this section we discuss preliminary properties of solutions for \eqref{1.111}. We found it
convenient to parametrise the vectors from the 
 trigonometric basis $\{e^{i{ \bk}\cdot x}\}$ by natural 
numbers and to normalise them.
 That is, to use the basis $\{e^j(x), j\ge 1 \}$, where 
 \be\label{basis}
 e^j(x) = (2\pi )^{-d/2} e^{i{ \bk}\cdot x},\qquad \bk=\bk(j).
 \ee
  The functions $e^j(x)$ 
are eigen--vectors of the Laplacian,  $-\Delta e^j=\lambda_j e^j$, so
ordered that $0=\lambda_1< \lambda_2\le \ldots$.  Accordingly eq.~\eqref{1.111} reads 
\begin{equation}\label{1.1}
\begin{split}
\dot u+i\nu^{-1}\big(-\Delta u\big)= -f(-\Delta) u- i\gi|u|^{2q_*}u
+\frac{d}{d\tau}\sum_{j=1}^\infty b_j\bb^j(\tau)e^j(x)\ ,
\end{split}
\end{equation}
$u=u(\tau,x)$, where 
$
f(-\Delta) e^j= \lla_je^j$ with $ \lla_j= f(\lambda_j). 
$
The processes $\bb^j = \beta^j + i\beta^{-j}  , j\ge 1$, are
standard independent complex Wiener processes. 
 The real numbers $b_j$ are such that  for a suitable sufficiently large even integer 
$r $ (defined below in \eqref{rr}) 
 we have 
$$
B_r:=2\sum_{j=1}^\infty \lambda_j^r
b_j^2 < \infty .
$$

By $\cH^p$, $p\in \R$, we denote the Sobolev space $\cH^p=H^p(\T^d, \C)$,
 regarded as a real Hilbert space,
 and denote by $\lan \cdot,\cdot \ran$ 
 the real $L^2$--scalar product on $\T^d$ . We provide $\cH^p$  with the  norm $\|\cdot \|_p$, 
$$
\left\| u\right\|_p^2=\sum_{j=1}^\infty|u_j|^2  ( \lambda_j\vee 1)^p\quad
\mbox{for } u(x)=\sum_{j=1}^\infty u_j e^j(x)\  .
$$

Let $u(t,x)$ be a solution of \eqref{1.1} such that
$u(0,x)=u_0$. It satisfies standard a-priori estimates which we now
discuss, following \cite{K12}. Firstly,  for a suitable $\eps_0>0$, uniformly in $\nu>0$
one has
\begin{equation}\label{1.3}
\E   e^{\eps_0 \|u(\tau)\|^2_0} \leq C(B_0,\|u_0\|_0) \quad \forall \tau\geq 0\,.
\end{equation}
 Assume that 
\begin{equation}\label{2.4}
q_*<\infty\;\;\;\text{if}\;\;d=1,2,\qquad
q_*<\frac{2}{d-2}\;\;\;\text{if}\;\;d\ge3\ . 
\end{equation}
Then, the following bounds on the Sobolev norms of the solution hold for each $2m\le r$ and every
 $n$:
\begin{equation}\label{2.5}
 \begin{split}
\E\left(\sup_{0\le\tau\le T}  {\|u(\tau)\|}_{2m}^{2n}+
\int_0^T{\|u(s)\|}^2_{2m+1} { \|u(s)\|}_{2m}^{2n-2}ds\right)\\ 
\le {\|u_0\|}_{2m}^{2n}+
C(m,n,T)\big(1+\|u_0\|_0^{c_{m,n}}\big),
\end{split}
\end{equation}
 \begin{equation}\label{2.05}
 \E\,
{\|u(\tau)\|}_{2m}^{2n} \le  C(m,n)\qquad  \forall\,\tau\ge0, 
\end{equation}
 where $C(m,n,T)$  and $C(m,n)$  also depend on
 $B_{2m}$.
 \medskip
 
  Estimates \eqref{2.5}, \eqref{2.05} are assumed everywhere in our work.  As we have explained, 
 they are fulfilled under the assumption  \eqref{2.4}, but if  the function $f(t)$ grows super-linearly, then the 
 restriction  \eqref{2.4} may be weakened.
  \medskip
 
 Relations \eqref{2.5} in the usual way
 (cf. \cite{Hai01b,KS04J, Od06, Sh06})
 imply that eq.~\eqref{1.1} is {\it regular in the space
   $\cH^{r}$} in the sense that for any $u_0\in
 \cH^{r}$ it has a unique strong solution $u(t,x)$,
 equal to $u_0$ at $t=0$, and satisfying estimates \eqref{1.3},
 \eqref{2.5} for any $n$. By the Bogolyubov-Krylov
 argument, applied to a solution of \eqref{1.1}, starting from the origin at $t=0$,
  this equation has a stationary measure $\mu^\nu$, supported
 by the space $\cH^{r}$, and a corresponding stationary
 solution $u^\nu(\tau)$, $\cD u^\nu(\tau)\equiv\mu^\nu$, also
 satisfies \eqref{1.3} and \eqref{2.05}.

\subsection{Resonant averaging  }\label{s3.1}
Let ${W}\in\Z^n$ , $n\ge1$, be a non-zero 
 integer vector such that its components are relatively 
prime (so if $W=mV$, where $m\in\Z$ and $V\in\Z^n$, then $m=\pm1$). 
We  call the set 
\begin{equation}\label{4.2}
\cA=\cA(W):=\{ s\in \Z^n:  \,{W}\cdot s=0\}\ 
\end{equation}
the {\it set of resonances for $W$}.  This is a $\Z$-module.
 Denote its rank by $r$.
Here and everywhere below the finite-dimensional vectors
are regarded as column-vectors.

\begin{lemma}\label{l1.1} The rank $r$ equals $n-1$. There exists a 
system $\zeta^1,\ldots,\zeta^n$ of integer vectors in
$\Z^n$ such that
span$\,{}_{\Z}\{\zeta^1,\ldots,\zeta^{n-1} \}= \cA$,  and the
$n\times n$ matrix $R=(\zeta^1 \zeta^2\dots\zeta^n)$
 is unimodular (i.e., $\det R=\pm1$).
\end{lemma}

That is, the vectors $(\zeta^1,\ldots,\zeta^{n-1} )$ make an integer basis of the 
hyperspace $W^\perp\subset\R^n$. 

\begin{proof}
We restrict ourselves to the case when some component of the vector $W$ equals one since 
this is the result we need below. For the general case and for a more general statement see,
 for example,  \cite{bour}, Section~7.
 
 Without loss  of generality we assume that $W_n=1$. Consider  the matrix such that 
 its $n$-th column is $W$ and for $j<n$ the $j$-th column is the vector $e^j=(e^j_1,\dots,e^j_n)^T$, 
 where $e^j_l=\delta_{j,l}$. It is unimodular and transforms the basis vector $e^n$ to $W$. Its
 inverse is an unimodular matrix $B$ such that $BW=e^n$. Let $s$ be any vector in $\cA$. Since
 $$
 W\cdot s=0 \Leftrightarrow  BW\cdot (B^T)^{-1}s=0 \Leftrightarrow 
 e^n\cdot  (B^T)^{-1}s=0 \,,
 $$
 then $(B^T)^{-1}s=\sum _{j=1}^{n-1} m_j e^j$, where $m_j$'s are some integers. 
  This proves the lemma if we choose $\zeta^j=B^Te^j$, 
 $j=1,\dots,n$.  Note that the matrix $R$ equals $B^T$.
\end{proof}
\medskip

Since $R^TW=BW=e^n$, then the automorphism of the torus $\T^n\to \T^n$,
$\vp\to y=R^T\vp$, ``resolves the resonances" in the differential equation 
$$
\dot \vp = W
$$
in the sense that it transforms it to the equation
\begin{equation}\label{resolve}
\dot y=R^TW=(0,\dots,0,1)^T\,.
\end{equation}
\medskip

Let us consider a  mapping $L=L_{\cA}:\T^n\to\T^{n-1}$, ``dual to the module $\cA$":
\begin{equation}\label{f0}
L: \T^n \ni\vp  \to ( \vp\cdot Re^1 ,\dots, \vp\cdot Re^{n-1} )^T  \in
\T^{n-1}\,.
\end{equation}
The basis $\{\eta^j=(R^T)^{-1}e^j,\ 1\le j\le n \}$, is dual to the basis  $\{\zeta^j=R e^j,\ 1\le j\le n \}$,
since 
$$
\eta^j\cdot\zeta^l = (R^T)^{-1}e^j\cdot Re^l =\delta_{j,l}.
$$
Therefore if we decompose $\vp\in\T^n$ in the $\eta$-basis, 
$\vp=\sum_ky_k\eta^k=(R^T)^{-1}y$, then $L\vp=(y_1,\dots, y_{n-1})^T$. That is, 
\begin{equation}\label{f1}
L\circ (R^T)^{-1} (y_1,\dots,y_n)^T=(y_1,\dots,y_{n-1})^T\,. 
\end{equation}
In particular, the fibers of the mapping $L$ are the circles  $R(\{{\mathbf y}\}\times S^1)$,
where ${\mathbf y}=(y_1,\dots,y_{n-1})^T  \in\T^{n-1}$. 
\medskip

 For a continuous 
function $f$ on $\T^n$ we define its {\it resonant average with respect to the integer vector $W$}  as the function
\begin{equation}\label{usred}
\langle f \rangle_W(\vp) :=
\int_0^{2\pi}f\left( \vp+ tW\right) \,\dbar t\ ,
\end {equation}
where we have set $\dbar t:= \tfrac{1}{2\pi}d t$.

\begin{lemma}\label{r.aver} Let $f$ be a  $C^\infty$-function on $\T^n$, 
 $f(\vp)=\sum f_se^{is\cdot \vp}$. Then
\begin{equation}\label{yy}
\lan f\ran_W(\vp)=
\sum
f_s\delta_{0,\,s\cdot W}\, e^{is\cdot \vp}=
\sum_{s\in \cA(W)} f_s\, e^{is\cdot \vp}.
\end{equation}
\end{lemma}

\noindent
{\it Proof.} It is immediate that \eqref{yy} holds for trigonometrical polynomial.
Since for  $C^\infty$-functions 
the series in \eqref{yy} converges well, then by continuity the result holds 
for smooth functions $f$.
\qed

\subsection{Resonant averaging in a  Hilbert space  }\label{s3.2}
Consider the Fourier transform for complex functions on $\T^d$ 
which we write as the
  mapping
$$
\cF: \cH \ni u(x)\mapsto v=(v_1,v_2,\ldots) \in \C^\infty\ ,
$$
defined by the relation $u(x)=\sum v_k e^k(x)$. In the space of
complex sequences we introduce the norms
\begin{equation}\label{vnorm}
\left|v\right|^2_{h^p}=\sum_{k\ge 1} |v_k|^2 (\lambda_k\vee1)^p\,, \quad p\in\R\,,
\end{equation}
and set $h^p=\{v|\, \left| v\right|_{h^p} <\infty\}$. Then 
$$
|\cF u|_{h^p}=\|u \|_p\qquad \forall\,p.
$$

For $k\ge1$ let us denote $I_k=I(v_k)=\tfrac12|v_k|^2$ and
$\vp_k=\vp(v_k)$, where for $v\in\C$ 
 $\vp(v)=\Arg v\in S^1$  if   $v\ne0$,  and  
$\vp(0)=0\in S^1$.  For any $r\ge0$ consider the mappings
\begin{equation}\label{phi} 
\Pi_I:h^r\ni v\mapsto I=(I_1,I_2,\dots)\in h^r_{I+},\quad
\Pi_\vp:h^r\ni v\mapsto \vp=(\vp_1,\vp_2,\dots)\in\T^\infty.
\end{equation}
Here $h^r_{I+}$ is the positive octant $\{I: I_k\ge 0 \  \forall k\}$ 
 in the space $h^r_{I}$, where 
$$
 h^r_{I}=\{I\mid | I|_{h^r_I}=2\sum_k (\lambda_k\vee1)^r|I_k|<\infty\}.
$$
Abusing a bit notation we will  write 
$
\Pi_I (\cF(u))=I(u)$,  $\Pi_\vp(\cF(u))=\vp(u).
$
The mapping $I:\cH^r\to h_I^r$ is  2-homogeneous 
continuous, while the mapping $\vp:\cH^r\to \T^\infty$ is  Borel-measurable
 (the torus 
$\T^\infty$ is given the Tikhonov topology and the corresponding   Borel sigma-algebra). 

For infinite integer vectors $s=(s_1,s_2,\dots)$  (and only for them) we will write
the $l_1$-norm of $s$ as $|s|$, 
$$
|s| = \sum_j|s_j|.
$$
We denote $
\Z_0^\infty=\{s\in\Z^\infty:  |s|<\infty\},
$
and for a vector $s=(s_1,s_2,\dots)\in \Z_0^\infty$ write
\begin{equation}\label{notation}
\Lambda\cdot s=\sum_k\lambda_ks_k, \quad \supp s=\{k: s_k\ne 0\}, \quad
\lc s\rc=\max\{k: s_k\ne0\}.
\end{equation}
Similar for  $\vp\in\T^\infty$ and $s\in \Z_0^\infty$ we write 
$\vp\cdot s=s\cdot\vp=\sum_k\vp_k s_k\in S^1$. 

Let us  fix some
 $m\in\N\cup\infty$ and  define the set  of resonances of order $m$
   for the (integer)  frequency-vector
 $\ 
 \Lambda=(\lambda_1,\lambda_2,\dots)
 $ 
  as
\begin{equation}\label{resset}
\cA(\Lambda,m)=\{ s\in \Z^\infty_0 : |s|\le m,\Lambda\cdot s=0\}\ .
\end{equation}
We will abbreviate 
 $\cA(\Lambda)= \cA(\Lambda,\infty)=\{ s\in \Z^\infty_0 : \Lambda\cdot s=0\}$. 

 Let us denote 
$\ 
\ZZ=\{s\in\Z^\infty_0: s_k\ge0\;\; \forall k\},
  $
and consider a series on some space $h^r, r\ge0$:
  \begin{equation}\label{xx}
  F(v)=\sum_{p,q,l\in\ZZ} \CC\,, 
  \end{equation}
  where $I=I(v)$,
  $C_{pql}=0$ if $\supp q\cap \supp l\ne\emptyset$
  and for $v\in h^r$, $q\in\ZZ$ we write $v^q=\prod v_j^{q_j}$. 
  We assume that the series  converges normally in $h^r$ in the sense that for each $R>0$ we have 
    \begin{equation}\label{xxx}
    \sum_{p,q,l\in\ZZ}  |C_{pql} | \sup_{|v|_{h^r}, |w|_{h^r}\le R} |v^pw^p v^qw^l|<\infty.
    \end{equation}
    Clearly $F(v)={\bold F}(v,\bar v)$, where ${\bold F}$ is a (complex) analytic function on $h^r\times h^r$.
    Abusing language and following a physical tradition we will say that 
    {\it $F$ is analytic in $v$ and $\bar v$}. In particular, $F(v)$ is a real-analytic (so 
     continuous) function of $v$, 
      and the series \eqref{xx}
     converges absolutely. 

The resonant averaging of $F$ can be conveniently  defined by introducing,  
for any $\theta\in\T^\infty$, the rotation operator $\Psi_\theta$,
which is a linear operator in $h^0$:
$$
\Psi_\theta(v)=v',\qquad v'_k=e^{i\theta_k}v_k.
$$
 Clearly this is an unitary isomorphism of every space $h^r$.
  Note  that 
   $
  (I\times\vp)(\Psi_\theta v)\equiv (I(v), \vp(v)+\theta)\,.
  $
  Using that $\Lambda$ is an integer vector and based on
  definition~\eqref{usred}, we 
give the following

\noindent
{\bf Definition.}  If a function $F \in C(h^r)$ is given by a normally
converging series \eqref{xx},  then its 
     resonant average with respect to $\Lambda$ is the function
  \begin{equation}\label{La_aver}
  \lan F\ran_\Lambda(v):= \int_0^{2\pi} F(\Psi_{t\Lambda} (v))\, \dbar t\,,
  \qquad \dbar t = dt/2\pi\,.
   \end {equation}
   
   Defining a function $\tilde F(I,\vp)$ by the relation 
   $
   F(v) = \tilde F(I(v), \vp(v)),
   $
   we see that 
   $ \lan F\ran_\Lambda(v) = \int_0^{2\pi} \tilde F (I, \vp+t\Lambda)\,\dbar t$. 
   So this definition well agrees with \eqref{usred}. 
   \smallskip

    Consider a monomial $F=(2I)^pv^q\bar v^l$.  By
    Lemma~\ref{r.aver} we have
     $$
     \lan (2I)^p v^q\bar v^l\ran_\Lambda= 
     (2I)^pv^q\bar v^l \delta_{0, (q-l)\cdot\Lambda}\ . 
     $$ 
     Now assume that  $F$ is given by a normally convergent series \eqref{xx} and has
     degree $\le m\le\infty$ in sense that 
     $C_{pql}=0$   unless  $ |q|+|l|\le m$. Then  
\begin{equation}\label{La_aver2}
    \lan F\ran_\Lambda(v) = \sum_{q-l\in\cA(\Lambda,m)} \CC=
    \sum_{(q-l) \cdot\Lambda=0} \CC\,.
\end{equation}
  If the series  \eqref{xx}   converges normally,  then the series in
  the r.h.s. above also does. It defines an analytic in $(v, \bar v)$ 
   function.
Note that in view of \eqref{La_aver2}
\begin{equation}\label{AAver}
\lan F\ran_\Lambda \text{ is a function of $I_1,I_2\dots$ and the variables 
$\{s\cdot \vp, s\in \cA(\Lambda,m)\}$}.
\end {equation}

\section{Averaging  for equation \eqref{1.1}.}\label{s5}
Everywhere below $T$ is a  fixed positive number.

\subsection{Equation \eqref{1.1} in the $v$-variables, resonant monomials and combinations of phases.
}\label{s5.1}

Let us pass in eq. \eqref{1.1} with $u \in\cH^r, \ r>d/2$,  to the $v$-variables,
$v=\cF(u)\in h^r$:
 \begin{equation}\label{5.1}
 \begin{split}
dv_k+i\nu^{-1}\lambda_kv_kd\tau=P_k(v)\,d\tau +\,
b_kd\bb^k(\tau),\quad k\ge1;\quad v(0)=\cF(u_0)=:v_0. 
\end{split}
\end{equation}
Here
\begin{equation}\label{5.00}
P_k=P_k^1+P_k^0,
\end{equation}
where $P^1$ and $P^0$ are, correspondingly, the linear and nonlinear  hamiltonian 
parts of the perturbation.  So $P^1_k$ is the Fourier-image of $-f(- \Delta)$, i.e. 
$P^1_k=\,$diag$\,\{-\lla_k,k\ge1\}$, while the  operator $P^0$ is the mapping 
$u\mapsto -i\rho|u|^{2q_*}u$, written in the $v$-variables. I.e., 
 $$
 P^0(v)=-i\gi \cF (|u|^{2q_*}u)\,, \;\; u=\cF^{-1}(v).
$$
Every its component $P^0_k$ is a sum of monomials:
\begin{equation}\label{P^0}
P_k^0(v)=\sum _{ {p,q,l\in\Z_{+0}^\infty }}  C_k^{pq l}(2I)^p
v^q{\bar v}^l = \sum _{ {p,q,l\in\Z_{+0}^\infty }}  P_k^{0pql}(v), 
\qquad k\ge1, 
\end{equation}
where 
$C_k^{pql}=0$  unless $2|p|+|q|+|l|=2q_*+1$ and $|q|=|l|+1$.
It is straightforward that  $P_k^0(I,\vp)$ (see \eqref{phi}) 
 is a function of $\vp=(\vp_j, j\ge1)$ of order $2q_*+1$, and that 
 the mapping $P^0$ is analytic of polynomial growth:

\begin{lemma}\label{l.P^0}
The nonlinearity $P^0$ defines a real-analytic transformation of 
$ h^r$  if  $r>\tfrac{d}{2}$.
 The mapping
 $P^0(v)$ and its differential $dP^0(v)$ both have  polynomial growth in $|v|_{h^r}$.
\end{lemma}
\smallskip

 We will refer to equations \eqref{5.1} as to the {\it
  $v$-equations.}

For any $s\in \Z^\infty_0$ consider the linear combination of phases
$$
\Phi^s:h^0\to S^1\ ,\quad v\mapsto s\cdot \vp(v)\,.
$$ 
We fix
 $$
 m=2q_*+2,
 $$
  and find the corresponding  set $\cA=\cA(\Lambda,m)$ of  resonances
  or order $m$   (see \eqref{resset}).  We order  vectors in  the set
  $\cA$,  that is   write it as   
$\cA=\{ s^{(1)},s^{(2)},\ldots\}$, 
  in such a way 
    that $\lc s^{(j_1)} \rc \le  \lc s^{(j_2)}\rc$ if $j_1\le j_2$,
    and for $N\ge1$ denote  
  \begin{equation}\label{NN}
 J(N)=\max\{j: \lc s^{(j)}\rc \le N\}.
\end{equation}
  For any $s^{(j)}\in \cA $ consider the corresponding resonant combination of phases $\vp(v)$,
$\Phi_j(v)= \Phi^{s_j}(v)$, 
and   introduce the Borel-measurable   mappings
\begin{equation*}
\begin{split}
&h^r\ni v\mapsto \Phi=(\Phi_1,\Phi_2,\dots)\in  S^1\times
S^1 \times \cdots=:\cT^\infty\ , \\
&h^r\ni v\mapsto (I\times \Phi)\in h^r_{I+}\times \cT^\infty\ . 
\end{split}
\end{equation*}
Note that the system $\Phi$ of resonant combinations is highly over-determined:
there are many linear relations between its components $\Phi_j$. 

Let us pass in eq. \eqref{5.1} from the complex variables $v_k$ to
the action-angle  variables
$I, \vp$:
 \begin{equation}\label{5.2}
dI_k(\tau)
=(v_k\cdot P_k)(v)\,d\tau + b_k^2\,d\tau  +b_k(v_k\cdot
d\bb^k)
\end{equation}
(here $\cdot$ indicates the real scalar product in  $\C\simeq\R^2$), 
 and 
 \begin{equation}\label{5.3}
 \begin{split}
d\vp_k(\tau)=\Big(-\nu^{-1}\lambda_k+
|v_k|^{-2} \big(iv_k\cdot P_k(v)\big)
\Big)\,d\tau+|v_k|^{-2}b_k(iv_k\cdot d\bb^k)\,.
\end{split}
\end{equation}
The equations for the actions are slow, while equations for the angles are fast since $d\vp_k \sim\nu^{-1}$. 
But 
the resonant combinations  $\Phi_j$ of  angles 
satisfy slow equations: 
\begin{equation}\label{5.ris}
 \begin{split}
d\Phi_j(\tau)=\sum_{k\ge 1}s^{(j)}_k\Bigl(|v_k|^{-2}(iv_k\cdot
P_k)\,d\tau+ |v_k|^{-2}b_k(iv_k\cdot d\bb^k)\Bigr), \quad  j\ge1. 
\end{split}
\end{equation}

Repeating for equations \eqref{5.1} and \eqref{5.2}
the argument from Section~7 in \cite{KP08} (also see Section~6.2 in \cite{K10}), we get low bounds for the norms of  the components
$v_k(\tau)$ of $v(\tau)$:

\begin{lemma}\label{l5.1} Let $v^\nu(\tau)$ be a solution of \eqref{5.1} and $I^\nu(\tau)=I(v^\nu(\tau))$. Then for 
 any $k\ge1$ the following convergence holds uniformly in $\nu>0$:
\begin{equation}\label{5.8}
\int_0^T\IP\{I^\nu_k(\tau)\le\delta\}\,d\tau\to0\qquad \text{as}\;\; \delta\to0
\end{equation}
(the rate of the  convergence  depends on $k$). 
\end{lemma}

Now we define and study corresponding resonant monomials of $v$. 
For any $s\in\Z_0^\infty$, vectors $s^+, s^-\in\ZZ$  such that $s=s^+-s^-$  and 
$
\supp s=\supp s^+\cup \supp s^-$,  $\supp s^+\cap \supp s^-=\emptyset$
are uniquely 
defined.
Denote by $V^s$ the monomial 
\begin{equation}\label{hren}
V^s(v)= v^{s^+}  \bar v^{s^-}=
\prod _l v_l^{s_l^+}   \prod_l \bar v_l^{s_l^-}.
\end{equation}
This is a real-analytic function on every space $h^l$, and $\vp \big(V^s(v)\big)=\Phi^s(v)$. 
{\it Resonant monomials} are the functions \footnote{It may be better to call $V_j(v)$ a  minimal resonant monomial since for any $l\in\ZZ$  the monomial $I^lV_j(v)$ also is resonant and corresponds to the same resonance.}

$$
V_j(v)= V^{s^{(j)}}(v),  \qquad j=1,2,\dots. 
$$
Clearly they satisfy
\begin{equation}\label{ugly}
I(V_j(v))=(2I)^{\frac12 |s^{(j) }| }:=\prod_l (2I_l)^{ \frac12 |s_l^{(j)}|}\,,\qquad 
\vp (V_j(v))=\Phi_j(v).
\end{equation}
Now  consider the mapping
\begin{equation}\label{reson}
\begin{split}
& V: h^l\ni v\mapsto (V_1,V_2,\dots)\in \C^\infty\ ,\\
\end{split}
\end{equation}
where $\C^\infty$ is given the Tikhonov topology. It is continuous for any $l$. For $N\ge1$ denote
$$
V^{(N)}(v)=\big(V_1,\dots,V_{J}(v)\big) \in \C^{J},
$$
where $J=J(N)$, see \eqref{NN}. 

For any $s\in\Z_0^\infty$, applying the  Ito formula to the process
 $V^s(v(\tau))$,  we get that 
 \begin{equation} \label{7.1}
 \begin{split}
 d\,V^s= V^s\Big(-i\nu^{-1}(\Lambda\cdot s) d\tau+ &\sum_{j\in\supp s^+} s_j^+
             v_j^{-1}(P_j(v)\,d\tau+b_j\,d\bb_j)\\
             +&\sum_{j\in\supp s^-} s_j^-\,
             {\bar  v_j}^{-1}(\bar P_j(v)\,d\tau+b_j\,d\bar\bb_j)\Big).
 \end{split}
 \end{equation}
 If $s=\tilde s\in\Z_0^\infty$ is perpendicular to $\Lambda$, 
  then the first term in the r.h.s. vanishes. So  $V^{\tilde s}(\tau)$ is a slow process,
  $dV^{\tilde s}\sim1$. In particular, the processes 
 $dV_j, j\ge1$, are slow. 
 
 Estimates \eqref{2.5} and equation \eqref{7.1} readily imply 
 
 \begin{lemma}\label{l7.1}
 For any $j\ge1$ we have 
 $\ 
 \E\big|V_j(v(\cdot))\big|_{C^{1/3}[0,T]} \le C_j(T)<\infty,
 $ 
 uniformly in $0<\nu\le1$. 
 \end{lemma}

Let us provide the space $C([0,T];\C^\infty)$ with the Tikhonov topology, identifying it with
the space $C([0,T];\C)^\infty$. This topology is metrisable by the Tikhonov distance.  From now on we 
fix an even integer $r$, 
\begin{equation} \label{rr}
r\ge \frac d2+1\,,
\end{equation} 
 and
abbreviate 
$$
h^r=h,\quad h^r_I=h_I,\quad C([0,T], h_{I+})\times 
C([0,T],\C^\infty)=:  \cH_{I,V}.
$$
We provide $\cH_{I,V}$ with Tikhonov's distance, the corresponding Borel $\sigma$-algebra and the natural
 filtration of the sigma-algebras
$\{\cF_t, 0\le t\le T\}$. 

Let us consider a solution $u^\nu(\tau)$ of eq. \eqref{1.1}, satisfying
$
u(0)=u_0,
$
denote $v^\nu(\tau)=\cF(u^\nu(\tau))$ and abbreviate
$$
I(v^\nu(\tau))=I^\nu(\tau),\quad V(v^\nu(\tau))= V^\nu(\tau)\in\C^\infty. 
$$

 \begin{lemma}\label{l7.2} 1) 
 Assume that $u_0\in\cH^r$. Then the 
  set of laws $\cD(I^\nu(\cdot), V^\nu(\cdot)), \  0<\nu\le1$, is tight in $\cH_{I,V}$. 
  
  2) Any limiting measure $\cQ$ for the set of laws in 1) satisfies
   \begin{equation} \label{apriori}
 \begin{split}
 \E^\cQ |I |^n_{C([0,T], h^r_I)} \le C_n &\quad \forall\, n\in\N,\qquad
 \E^\cQ \int_0^T|I(\tau)|_{h_I^{r+1}}d\tau \le C',\\
 & \E^\cQ e^{\e_0 |I(\tau)|_{h^0_I}} \le C^{''} \quad \forall\,\tau\in[0,T]. 
 \end{split}
 \end{equation}

  \end{lemma}
  \begin{proof} 1) 
  Due to Lemma \ref{l7.1} and the Arzel\`a Theorem, the laws of processes $V_j(v^\nu(\cdot))$,
  $0<\nu\le1$, are tight in $C([0,T],\C)$, for any $j$. Due to estimates \eqref{2.5} with $n=1$
   and since the 
  actions $I_k^\nu$ satisfy slow equations \eqref{5.2}, the laws of processes $I^\nu(\tau)$ are 
  tight in $C([0,T], h_{I+})$ (e.g. see in \cite{VF}).
  Therefore, for every $N$, 
   any sequence $\nu_\ell\to0$ contains a subsequence 
  such that the laws $\cD\big(I^\nu(\cdot), V^{(N)}(v^\nu(\cdot))\big)$ converges along it to
  a limit.  Applying the diagonal process we get another subsequence $\nu'_\ell$ such 
  that the convergence holds for each $N$. The corresponding limit is a measure $m^N$ on the 
  space 
  $
  C([0,T], h_{I+})\times C([0,T],\C)^{J(N)}.
  $
  Different measures $m^N$ agree, so by Kolmogorov's theorem they correspond to some measure $m$ on the 
  sigma-algebra, generated by cylindric subsets of the space
  $
  C([0,T], h_{I+})\times C([0,T],\C)^\infty,
  $
  which coincides with the Borel sigma-algebra for that space. 
  It is not hard to check that 
  $
  \cD(I^\nu(\cdot), V^\nu(\cdot))
  \strela m
  $
  as $\nu=\nu'_\ell\to0$. 
  This proves the first  assertion. 
  
  2) Estimates \eqref{apriori} follow from \eqref{1.3}, \eqref{2.5}, the weak convergence to $\cQ$
  and the Fatou lemma; cf. Lemma~1.2.17 in \cite{KS}. 
  \end{proof}

\subsection{Averaged equations, effective equation,  interaction representation }
\label{s2.3}
  Fix  $u_0\in \cH^r$ and consider  any limiting measure $\cQ^0$ 
  for the laws 
\begin{equation} \label{5.88}
\cD(I^{\nu_\ell}(\cdot), V^{\nu_\ell}(\cdot))\strela \cQ^0
\as \nu_\ell\to0,
\end{equation}
existing by Lemma \ref{l7.2}. 
Our goal is to show that the limit $\cQ^0$ does not depend on the
sequence $\nu_\ell \to0$
and develop tools for its study. We begin with writing down averaged
equations for the slow components $I$ and $\Phi$ of the process $v(\tau)$,
using the  rules of the stochastic calculus (see \cite{Khas68, FW03}), and 
formally replacing there the  usual averaging in $\vp$ by the resonant 
averaging $\lan\cdot\ran_\Lambda$.
Let us first consider the $I$-equations \eqref{5.2}. The drift in the $k$-th equation is
$$
b_k^2+v_k\cdot P_k=b_k^2+v_k\cdot P_k^1+ v_k\cdot P_k^0,
$$
where $v_k\cdot P_k^1=- 2\lla_k I_k$ and 
$
v_k\cdot P_k^0(v) =  \sum _{ {p,q,l\in\Z_{+0}^\infty }} v_k\cdot  P_k^{0pql}(v),  
$
see \eqref{P^0}. By Section \ref{s2.5} the sum converges normally, so the resonant averaging of the drift
is well defined. The dispersion matrix for eq.~\eqref{5.2} with respect to the real Wiener processes 
$(\beta^1,\beta^{-1}, \beta^2,\dots)$ is diag$\,\{b_k(\Re v_k \,\Im v_k), k\ge1\}$ (it is formed by $1\times2$-blocks). 
The diffusion matrix equals the 
dispersion matrix times its conjugated  and equals diag$\,\{b_k^2|v_k|^2, k\ge1\}$. It is independent 
from the angles, so the averaging does not change it. 
 For its  square-root we take diag$\,\{b_k \sqrt{2I_k} \}$, and accordingly 
write the $\Lambda$-averaged $I$-equations as 
\begin{equation}\label{5.21}
dI_k(\tau) =
\langle v_k\cdot
P_k\rangle_{\Lambda}(I,V)\,d\tau + 
b_k^2\,d\tau  +{b_k}\sqrt{2I_k}\, d\beta^k(\tau), \quad  k\ge 1
\end{equation}
(see \eqref{AAver}). 

Now consider equations \eqref{5.ris} for  resonant combinations $\Phi_j$ of the angles. The corresponding 
dispersion  matrix  $D=(D_{jk})$ is formed by $1\times2$-blocks 
$$
D_{jk} = -s_k^{(j)} b_k(2I_k)^{-1} (\Im v_k  \, -\Re v_k).
$$
Again the diffusion matrix does not depend on the angles and equals $M=(M_{j_1 j_2})$,
$M_{j_1 j_2} =\sum_k s_k^{(j_1)}s_k^{(j_2)} b_k^2 (2I_k)^{-1}$. The matrix $D^{new}$ with the entries 
$D^{new}_{jk}=s_k^{(j)}b_k(2I_k)^{-1/2}$ satisfies $|D^{new}|^2=M$, and we write the averaged equations for 
$\Phi_j$'s as 
\begin{equation}\label{5.22}
d{\Phi_j}(\tau)= 
 \sum_{k \ge
  1}s^{(j)}_k\Bigl(\frac{\langle i v_k\cdot 
  P_k\rangle_{\Lambda}(I,V)}{ 2I_k}\,d\tau +\frac{{b_k}}{\sqrt{2I_k}}\,
d\beta^{-k}(\tau) \Bigr)\ , \quad  j\ge 1
\end{equation}
(we  use here Wiener processes, independent from those in eq.~\eqref{5.21}
since the differentials  $v_k\cdot d\bb^k$ and $iv_k\cdot d\bb^k$, corresponding to the noises in equations 
\eqref{5.2} and \eqref{5.3}, are independent).

Equations \eqref{5.21}, \eqref{5.22}  is a system of stochastic differential equations for the process $(I,V)(\tau)$
since each $\Phi_j$ is a function of $I$ and $V_j$. It is over-determined as there are linear relations
between various $\Phi_j$'s. Besides,  eq. \eqref{5.21} has a weak singularity at the locus 
$\ 
\Game (h)= \cup_k \{v\in h:  v_k=0\},
$
while  eq.~\eqref{5.22} has there a strong singularity. 
\medskip

 Consider  a component  $\lan v_k\cdot P_k^0\ran_\Lambda(v)$
 of the averaged drift in the equation for $I_k$. 
 It may be written as
\begin{equation}\label{xa}
\lan v_k\cdot P_k^0\ran_\Lambda(v)=  \int_0^{2\pi} v_k \cdot \Big ( e^{-i t\lambda_k}
  P_k^{0}(\Psi_{t\Lambda}  (v) )
  \Big)\,\dbar t  = v_k\cdot 
 R^{0}_k(v)\ ,
\end{equation}
where we set 
$\ 
 R^{0}_k(v)=  \int_0^{2\pi} e^{-i t\lambda_k}
  P_k^{0}(\Psi_{ t \Lambda }  (v) )\,\dbar t \ .
$ 
That is, 
\begin{equation}\label{xi}
 R^{0}(v)=  
 \int_0^{2\pi}  \Psi_{ - t\Lambda} 
  P^{0}(\Psi_{ t\Lambda}  v )\,\dbar t\ .
\end{equation}
Repeating the derivation of \eqref{La_aver2} and using that $|q|+|l|\le m-1$, 
 we see that 
\begin{equation} \label{ura}
R^0_k(v) = \sum_{\substack{p,q,l\in\ZZ \\ q-l\in\cA(\Lambda,m)+e^k\\ |q|+|l|+1\le m}} 
C_k^{pql}  (2I)^pv^q\bar v^l. 
\end{equation}

The relation \eqref{ura} interprets $R^0(v)$ as a sum of resonant
terms of the mapping $P^0(v)$, very much in the spirit of the WT,
while \eqref{xi} interpret it a result of the resonant averaging  of
$P^0$.

The vector field $R^0$ defines locally-Lipschitz operators in the spaces $h^p$, $p>d/2$:
 \begin{equation}\label{p4}
|R^0(v)-R^0(w)|_{h^p}\le C_p\big( |v|_{h^p}\vee |w|_{h^p}\big)^{2q_*}|v-w|_{h^p}.
  \end{equation}
  Indeed, in view of \eqref{xi}, for any $v,w$ such that $|v|_{h^p},
  |w|_{h^p}\le R$ we have
    \begin{equation}\label{p44}
    \begin{split}
|  (
 R^{0}(v) &- R^0(w))|_{h^p}  \le 
 \int_0^{2\pi} \Big
 | \Psi_{ - t\Lambda} \big(
  P^{0}(\Psi_{ t\Lambda}  v ) - P^0(\Psi_{ t\Lambda}  w )\big) \Big|_{h^p}
  \,\dbar t \ .
  \end{split}
   \end{equation}
   Since $P^0(v)=-i\rho\cF(|\hat v|^{2q_*}\hat v)$, where $\hat
   v=\cF^{-1}v$, then  denoting 
  $  \Psi_{ t\Lambda}  v =v_t$, defining $w_t$ similarly
   and using that the operators $\Psi_\theta$ define isometries of
   $h^p$, we bound the r.h.s.  
  of \eqref{p44}   by
   \begin{equation*}
   \begin{split}
    \int_0^{2\pi} \big| P^0(v_t) - P^0(w_t)\big|_{h^p}\,\dbar t =
    \rho  \int_0^{2\pi}  \big\||  \widehat {v_t}|^{2q_*} \widehat{ v_t} - 
    | \widehat{ w_t}|^{2q_*} \widehat{ w_t}   \big\|_p \,\dbar t
   \\
  \le \rho C_p  R^{2q_*}   \int_0^{2\pi}  \|\widehat{v_t} - \widehat{w_t}\|_p \,\dbar t
   \le \rho C_p R^{2q_*}|v-w|_{h^p}.
  \end{split}
  \end{equation*}
  \medskip

Finally  we set 
$$
R=R^0+R^1,\quad \text{where}\quad
 R^1_k(v)=P^1_k(v)=-\lla_k v_k. 
$$
Since 
$\ 
\lan v_k\cdot P^1_k\ran_\Lambda=\lan-\sum 2\lla_kI_k\ran_\Lambda=v_k\cdot P^1_k
=v_k\cdot R^1_k,
$
then in view of \eqref{xa} we have 
\begin{equation}\label{R_k}
\lan v_k\cdot P_k\ran_\Lambda(v)=v_k\cdot R_k(v). 
\end{equation}
For further usage we note that by the same argument, 
$\lan iv_k\cdot P_k^0\ran_\Lambda= iv_k\cdot R_k^0$ and 
$\lan iv_k\cdot P_k^1\ran_\Lambda=0= iv_k\cdot R_k^1$. So also 
\begin{equation}\label{R_kk}
\lan i v_k\cdot P_k\ran_\Lambda(v)=iv_k\cdot R_k(v). 
\end{equation}

Motivated by the averaging theory for equations without  resonances  in \cite{K10, K12}, we now
consider the following {\it  effective equation} for the slow dynamics in
 eq.~\eqref{5.2}:
\begin{equation}\label{5.eff}
d v_k= R_k(v) d \tau +{b_k}d \bb^k\ ,
\qquad k\ge1\,.
\end{equation}
In difference  with the averaged equations \eqref{5.21} and \eqref{5.22}, the effective equation is  regular,
i.e. it does not have  singularities at the locus 
 $\Game(h)$.  Since  $R^0: h\to h$ is locally Lipschitz,  then strong solutions for \eqref{5.eff} exist locally in
  time and are unique:
 
  \begin{lemma}\label{l.uniq} 
  A strong solution of eq. \eqref{5.eff} with a specified initial data $v(0)=v_0\in h$ is unique, a.s.
 \end{lemma}

The relevance  of the effective equation  for the study of 
the long-time dynamics in equations \eqref{1.1}=\eqref{5.1} is 
clear from the next lemma:
 
 \begin{lemma}\label{l.ef_eq}
 Let a continuous process $v(\tau)\in h$ be a weak solution of \eqref{5.eff} 
 such that all moments of the random variable $\max_{0\le\tau\le T} |v(\tau)|_h$
 are finite. 
 Then $I(v(\tau))$ is a 
 weak solution of \eqref{5.21}.
 Let stopping times $0\le\tau_1<\tau_2\le T$ and numbers 
 $\delta_*>0, N\in\N$ be such that
 \begin{equation}\label{stop}
 I_k(v(\tau)) \ge\delta_*\quad \text{for $\tau_1\le\tau\le\tau_2$ and $k\le N$.}
 \end{equation}
 Then the process 
 $\big(I(v(\tau)), \Phi_j(v(\tau)), j\le J(N)\big)$ is a weak solution of  the system of  averaged 
 equations \footnote{This system is heavily under-determined. }
  \eqref{5.21}, \eqref{5.22}${}_{j\le J}$.
 \end{lemma}
 \begin{proof} Let $v(\tau)$ satisfies \eqref{5.eff}. 
  Applying Ito's formula    to $I_k(v(\tau))$ and $\Phi_j(v(\tau))$, $ j\le J$, we get that 
  \begin{equation}\label{dIk}
 dI_k = v_k\cdot R_k\,d\tau + b_k^2\,d\tau +b_kv_k\cdot d\bb^k
 \end{equation}
 and 
 \begin{equation*}
 d\Phi_j =\sum_{k\in \supp s^{(j)}} s^{(j)}_k \left(\frac{iv_k\cdot R_k}{|v_k|^2}d\tau+ 
  \frac{b_k }{|v_k|^2} iv_k\cdot d\bb^k\right).
\end{equation*}
Using \eqref{R_k} and \eqref{R_kk} 
we see that \eqref{dIk} has the same drift and diffusion as \eqref{5.21}. 
So $I(v(\tau))$ is a weak solution of  \eqref{5.21} (see \cite{Yor74, MR99}). Similar, for 
$\tau\in[\tau_1,\tau_2]$, in view of \eqref{R_kk}, the process $(I,\Phi_j, j\le J)$, is a weak solution
of  the system  \eqref{5.21}, \eqref{5.22}${}_{j\le J}$.   \end{proof}

Now we show that the  effective equation describes the limiting (as $\nu\to0$) dynamics 
for the equations of motions, written 
in the $\ai$-variables of the interaction
representation \eqref{a}. Indeed, let  $u^\nu(\tau)$ be  a solution
of eq. \eqref{1.1}, satisfying $u(0)=u_0$. Denote
$v^\nu(\tau)=\cF(u^\nu(\tau))$ and consider the vector of  $\ai$-variables
$\ai^\nu(\tau) = (\ai^\nu_k(\tau)= e^{i\nu^{-1} \lambda_k
  \tau}  v^\nu_k(\tau),\, k\ge1)$ (cf. \eqref{a}). 
Notice that we obviously have
\begin{equation}\label{newesti}
|v^\nu(\tau)|_{h^m}\equiv |  \ai^\nu(\tau) |_{h^m}\;\; \forall\,m, \quad
I(v^\nu(\tau))\equiv I(\ai^\nu(\tau)),\quad
V(v^\nu(\tau))\equiv V(\ai^\nu(\tau))
\end{equation}
(see \eqref{reson}).  From
\eqref{5.1} we 
obtain the following system of equations for the  vector $\ai^\nu(\tau)$:
\begin{equation*} 
 \begin{split}
d\ai^\nu_k= \left(R_k(\ai^\nu) + \cR_k(\ai^\nu,\nu^{-1}\tau)\right)\,d\tau +\,
b_ke^{i\nu^{-1} \lambda_k \tau}d\bb^k(\tau), \quad k\ge1\ , \\
\end{split}
\end{equation*}
where we have denoted 
\begin{equation}\label{eq:nonres}
\cR_k(\ai,\nu^{-1}\tau)= \sum_{\substack{p,q,l\in\ZZ
    \\ q-l -e^k\not \in\cA(\Lambda,m)\\ |q|+|l|+1\le m}}  
P^{0pql}_k(\ai)\exp\Bigl(-i\nu^{-1}\tau\left( 
\Lambda \cdot (q- l- e^k)\right)\Bigr)\ .
\end{equation}
This  is the nonresonant, fast oscillating part of the nonlinearity (because 
$|\Lambda\cdot(q-l-e^k)|\ge1$).
Since $\{\bar \bb^k(\tau):=\int e^{i\nu^{-1} \lambda_k
  \tau}d\bb^k(\tau),\ k\ge1\}$ is another set of standard independent complex Wiener processes,
  then the process $\ai^\nu(\tau)$ is a weak solution of the system of equations
\begin{equation}\label{5.21a}
 \begin{split}
d\ai^\nu_k=&\left(R_k(\ai^\nu)+\cR_k(\ai^\nu,\nu^{-1}\tau)\right)\,d\tau +\,
b_kd \bb^k(\tau)\,, \quad k\ge1\ .\\
\end{split}
\end{equation}
 We will refer to equations \eqref{5.21a} as to the {\it
  $\ai$-equations.} It is crucial  that they 
   are identical to
the effective equation \eqref{5.eff}, apart from terms which
oscillate fast  as $\nu\to 0$.

 \subsection{Properties of resonant Hamiltonian $\cHR$ and   effective equation}
   \begin{lemma}\label{l.ham}
The vector field
$R^0$  is hamiltonian:
\begin{equation}\label{p2}
R^0=i \rho\nabla \cH^{\text{res}}(v),\quad \forall\,v\in h^p,\; p>d/2, 
\end{equation}
where $\cH^{\text{res}}(v) = \lan \cH\ran_\Lambda(v)$  and $\cH$ is the Hamiltonian  \eqref{*ham}.
   \end{lemma}
\begin{proof}     Indeed, 
since $P^0(v)=i\rho\nabla \cH(v)$, then 
  \begin{equation*}
\begin{split}
 R^{0}(v)=  \int_0^{2\pi}   \Psi_{- t\Lambda}   \Big(
  i\rho \nabla
  \cH
  (\Psi_{ t\Lambda}  (v) )\Big)\,\dbar t 
  = i\rho \nabla_{v} \int_0^{2\pi}  \cH
  ( \Psi_{t\Lambda}  (v) )\,\dbar t
  =i\rho \nabla_{v}  \cH^{\text{res}}(v),
  \end{split}
\end{equation*}
as $\Psi_\theta^*\equiv\Psi_{-\theta}$, and where we used \eqref{xi}. 
\end{proof}

 Clearly $\cH^{\text{res}}(0)=0$.  Since $\cH(u)\le -C \|u\|_0^{2q_*+2}$ by the H\"older 
 inequality and since the transformations $\Psi_{t\Lambda}$ preserve $\|u\|_0$,  then
 $$
 \cH^{\text{res}}(u)\le -C\|u\|_0^{2q_*+2}\quad \forall\,u\,. 
 $$
 
 The resonant Hamiltonian $\cH^{\text{res}}$  has symmetries, given by some rotations $\Psi_m, m\in\R^\infty$:

  \begin{lemma}\label{l.symm}
  i) Let  ${\mathbf 1}=(1,1,\dots)$.  Then $\cHR( \Psi_{t{\mathbf 1}}v)=\,$const  (i.e., it does
  not depend on $t$); 
  
  ii) Let  $\cM^l$  the 
  $l$-th component of the sequence $(\bk(1), \bk(2),\dots)$, $l=1,\dots,d$ 
   (see \eqref{basis}).  Then $\cHR( \Psi_{t\cM^l}  \,v)=\,$const, for each  $l$.

  iii) $\cHR(\Psi_{t\Lambda}v)=\,$const.
   \end{lemma}
\begin{proof} 
i) By \eqref{La_aver} we  have 
$$
\cHR(\Psi_{t{\mathbf 1}}v)=
\int_0^{2\pi} \cH\big( \Psi_{t'\Lambda}(\Psi_{t{\mathbf
    1}}v)\big)\,\dbar t'=
\int_0^{2\pi} \cH\big(  \Psi_{t{\mathbf 1}}(
\Psi_{t' \Lambda}v) 
\big)\,\dbar t'\ .
$$
Let us denote  $\Psi_{t{\mathbf 1}}(\Psi_{t'\Lambda}v)  = v(t;t')$. Then 
$
(d/dt) v(t;t')=iv.
$
The flow of this hamiltonian  equation commutes with that of the equation with the Hamiltonian $\cH$.\footnote{
This follows from the fact that the functional $\tfrac12|v|^2_{h^0}$ is an integral of motion for the Hamiltonian $\cH$, which 
becomes obvious if we note that in the $u$-representation $\cH$ has the form \eqref{*ham} and $\tfrac12|v|^2_{h^0}$ is $\,\tfrac12\int |u|^2(x)\,dx$. }
So
$\cH(v(t;t'))$ is independent from $t$ for each $t'$, and i) follows since
$\cHR(\Psi_{t{\mathbf 1}}v)=\int \cH(v(t;t'))\,\dbar t'$. 

ii) Proof is the same since the transformations $\Psi_{t\cM^l}, t\in\R$, are the flow of the momentum
Hamiltonian 
$M^l(u)=\frac12 \sum_{j=1}^\infty \bk^l(j)|u_j|^2$,
which commutes with $\cH$.

iii) It is a straightforward consequence of the definition
\eqref{La_aver} of the resonant averaging.
  \end{proof}

  Since the transformations $ \Psi_{t{\mathbf 1}}$ form the flow of the Hamiltonian $H_0(v)=\tfrac12\sum|v_j|^2 = \tfrac12 |v|^2_{h^0}$,  the transformations $\Psi_{t\Lambda}$ -- the flow of  $ H_1(v)=\tfrac12\sum \lambda_j|v_j|^2$, and 
  the transformations  $\Psi_{t\cM^l}, t\in\R$ --  the flow of the momentum
Hamiltonian, we may  recast the assertions of the last lemma as follows:
  \begin{equation}\label{integrals}
  \{\cHR,H_0\}=0,\quad  \{\cHR,H_1\}=0,  \quad \{\cHR,\cM^l\}=0\;\; \forall\,l. 
    \end{equation}
  Here $\{\cdot,\cdot\}$ signifies the Poisson bracket. 
  As   the transformations $\Psi_{m},\ m\in\R^\infty$, are symplectic, then the symmetries in the lemma   above  preserve the   hamiltonian vector field $R^0$ and commute with it.   
    In particular, since 
  $\Psi_{t\Lambda}=e^{-it\Delta}$, then the spectral spaces $E_\lambda$ of the operator $-\Delta$,
  $$
  E_\lambda=\text{span}\, \{e^j: \lambda_j=\lambda\}\,,
  $$
  are invariant for the flow-maps of $R^0$.   
\medskip

Since the transformations  $\Psi_{m},m\in\R^\infty$, obviously preserve the vector field $R^1$ as well as  the law of 
the random force in \eqref{5.eff} (see the proof of the lemma below), 
then those  $\Psi_m$ which are symmetries of $R^0$ (equivalently,  which are 
symmetries of the Hamiltonian $\cHR$), preserve weak solutions of \eqref{5.eff}. So we have:

  \begin{lemma}\label{l.invar}
  If $v(\tau)$ is a solution of equation \eqref{5.eff} and $m\in\R^\infty$ be either a vector $m=t{\mathbf 1}, t\in\R$,
   or a vector $m=t\Lambda$, or  $m=t\cM^l$, $l=1,\dots,d$, 
  then $\Psi_{m}v(\tau)$ also is a weak solution.
   \end{lemma}
   \begin{proof}
   Denote $\Psi_{m}v(\tau)=v'(\tau)$. Applying $\Psi_{m}$ to eq.~\eqref{5.eff}, using Lemma~\ref{l.symm}
   and exploiting 
   the invariance of the operator $R^1$ with respect to $\Psi_{m}$,  we get
 $$
 dv'_k = \big(\Psi_{m} R(v(\tau)\big)_kd\tau + e^{im_k} b_kd\bb^k= (R(v'(\tau))_k+b_k( e^{im_k} d\bb^k).
 $$
 Since $\{e^{im_k} \bb^k(\tau), k\ge1\}$ is another set  of standard independent Wiener processes, then 
 $v'(\tau)$ is a weak solution of \eqref{5.eff}.
   \end{proof}
   
    \begin{corollary}\label{c.invar}
    If $\mu$ is a stationary measure for equation \eqref{5.eff} and a vector $m$ is as in
    Lemma~\ref{l.invar}, then the measure $\Psi_{m}\circ\mu$ also is stationary. 
       \end{corollary}
       
    The next lemma characterises   the increments of  $R^0(v)$ in the space $h^0$.
    It  will be needed below  to study the ergodic properties of the effective equation:

\begin{lemma}\label{estil2}
 Let $p>d/2$. 
Then  for any $v_1,v_2\in h^p$ we have 
 \begin{equation*}
|R^0(v)-R^0(w)|_{h^0}\le C \big( |v|_{h^p}+ |w|_{h^p} \big)^{2q_*} |v-w|_{h^0}.
  \end{equation*}
  \end{lemma}
  \begin{proof}
  Repeating the proof of the Lipschitz property of $R^0$ in the space $h$ (see \eqref{p4})
  and using the notation of that proof, i.e. denoting  $  \Psi_{
    t\Lambda}  v =v_t\,$, 
   $\hat v=\cF^{-1}v$, and similar for the vector $w$,  we get  that
      \begin{equation*}
    \begin{split}
|  &
 R^{0}(v) - R^0(w)|_{h^0} \le 
 \int_0^{2\pi} \Big
 | \Psi_{ - t\Lambda} \big(
  P^{0}(\Psi_{ t\Lambda}  v ) - P^0(\Psi_{t\Lambda}  w )\big) \Big|_{h^0}
  \,\dbar t\\
  &= 
    \int_0^{2\pi} \big| P^0(v_t) - P^0(w_t)\big|_{h^0}\,\dbar t =
      \int_0^{2\pi}  \big\||  \widehat {v_t}|^{2q_*} \widehat{ v_t} - 
    | \widehat{ w_t}|^{2q_*} \widehat{ w_t}   \big\|_0 \,\dbar t
   \\
 & \le  C  \int_0^{2\pi} ( |\widehat{v_t}|_{L^\infty} + |\widehat{w_t}|_{L^\infty})^{2q^*}      
    \|\widehat{v_t} - \widehat{w_t}\|_0 \,\dbar t
   \le  C_1 (|v|_{h^p}+|w|_{h^p})^{2q^*} |v-w|_{h^0}\ .
  \end{split}
  \end{equation*}
  \end{proof}

\section{Explicit calculation}\label{s2.5}
We intend here to calculate explicitly the  effective equation
\eqref{5.eff}, keeping  track of the dependence on the size $L$ of the
torus. To do that, it is
convenient to use the natural parametrisation of the exponential basis 
by  vectors $\bk\in\Z^d_L$; that is, decompose functions $u(x)$ to
Fourier series,  
$
u(x)= \sum_{\bk\in \Z^d_L}v_\bk e^{i \bk\cdot x}\  .
$
We modify the norms $|\cdot |_{h^p}$ accordingly :
$$
\left\| u\right\|^2_p =(2\pi L)^d\sum_{\bk\in\Z^d_L}\left(
|\bk|\vee \frac{1}{L}\right)^{2p} |v_\bk|^2=:
\left|v\right|^2_{h^p}\ .
$$
Now, as in the Introduction, the eigenvalues of the minus-Laplacian are
$\lambda_\bk=|\bk|^2 $ and the damping coefficients
$\gamma_\bk=f(\lambda_\bk)$.

In the $v$-coordinates the nonlinearity becomes the mapping $v\mapsto
P^0(v)$, whose $\bk$-th component is
$$
P^0_\bk(v)=-i\gi \sum_{\bk_1,\ldots \bk_{2q_*+1}\in \Z^d_L}
v_{\bk_1} \cdots  v_{\bk_{q_*+1}} \bar v_{\bk_{q_*+2}}\cdots \bar v_{\bk_{2q_*+1}}
\delta^{1\ldots q_*+1}_{q_*+2\ldots 2q_*+1\, \bk}\ 
$$
(see \eqref{N1}). Accordingly, 
\begin{equation}\label{eq:example}
v_\bk\cdot P^0_\bk
= \gi \sum_{\bk_1,\ldots \bk_{2q_*+1}\in \Z^d_L} \Im\,(v_{\bk_1} \cdots
v_{\bk_{q_*+1}} \bar v_{\bk_{q_*+2}}\cdots \bar v_{\bk_{2q_*+1}} 
\bar v_\bk )
\delta^{1\ldots q_*+1}_{q_*+2\ldots 2q_*+1\, \bk}\,.
\end{equation}
In order to calculate the resonant average, we first  notice that $v_\bk\cdot
P^0_\bk$ can be written as a series \eqref{xx}, where $|C_{pql}| \le 1$ and
$|q|+|p|+|l|=2q_*+2$. In this case the sum
in the l.h.s. of \eqref{xxx} is bounded by
$$
C \left( \sum_{\bk\in\Z^d_L}|v_\bk|\right)^{2q_*+2}\le C_1(L)|v|_p^{q_*+1} \left(
\sum_{\bk\in\Z^d_L}|\bk|^{-2p} \right)^{q_*+1}.
$$
So the condition \eqref{xxx} is met if $2p>d$. 

Since the order of the resonance $m=2q_*+2$, then  $\lan v_\bk\cdot P^0_\bk\ran_\Lambda(v)$
equals 
$$
 \gi \sum_{\bk_1,\ldots \bk_{2q_*+1}\in \Z^d_L}
\Im\,(v_{\bk_1} \cdots  v_{\bk_{q_*+1}} \bar v_{\bk_{q_*+2}}\cdots \bar
v_{\bk_{2q_*+1}}\bar v_\bk)
\delta^{1\ldots q_*+1}_{q_*+2\ldots 2q_*+1\, \bk}\delta(\lambda^{1\ldots
  q_*+1}_{q_*+2\ldots 2q_*+1\, \bk})\ ,
$$
(see \eqref{N2}).
This  follows from \eqref{eq:example} and
\eqref{La_aver} if one notes  that  appearing  there
 restriction
$(q-l)\cdot \Lambda=0$ is now  replaced 
by
 the factor $\delta(\lambda^{1\ldots q_*+1}_{q_*+2\ldots 2q_*+1\, \bk})$.
In a similar way, we see that the quantity $R^0_k$ , entering  equation \eqref{5.eff},  takes the form
$$
R^0_\bk(v)= -i\gi \sum_{\bk_1,\ldots \bk_{2q_*+1}\in \Z^d_L}
v_{\bk_1} \cdots  v_{\bk_{q_*+1}} \bar v_{\bk_{q_*+2}}\cdots \bar
v_{\bk_{2q_*+1}}
\delta^{1\ldots q_*+1}_{q_*+2\ldots 2q_*+1\, \bk}\delta(\lambda^{1\ldots
  q_*+1}_{q_*+2\ldots 2q_*+1\, \bk})\ .
$$
Taking into account that $R^1_\bk=-\lla_\bk v_\bk$, we finally arrive at an explicit formula
for the effective equation 
\eqref{5.eff}:
 \begin{equation}\label{explicit}
\begin{split}
&dv_\bk= \Bigl(-\lla_\bk v_\bk \\
&-i\gi \sum_{\bk_1,\ldots \bk_{2q_*+1}\in \Z^d_L}
v_{\bk_1} \cdots  v_{\bk_{q_*+1}} \bar v_{\bk_{q_*+2}}\cdots \bar
v_{\bk_{2q_*+1}}
\delta^{1\ldots q_*+1}_{q_*+2\ldots 2q_*+1\, \bk}\delta(\lambda^{1\ldots
  q_*+1}_{q_*+2\ldots 2q_*+1\, \bk})\Bigr)d\tau \\
&{} \qquad  \qquad \qquad \qquad  \qquad \qquad\qquad  
\qquad  \qquad \qquad
+b_\bk d\bb^\bk\ , \qquad\quad \bk\in \Z^d_L\ .
\end{split}
\end{equation}
Due to \eqref{p2},
$$
R^0_k(v) = i\rho \nabla_{v_k}\cH^{\text{res}}(v)=
 2i\rho \frac{\p}{\p \bar v_k}
 \cH^{\text{res}}(v).
$$
Therefore eq. \eqref{explicit}  can be written as the damped--driven hamiltonian
system  \eqref{*eff1}.
\medskip

\noindent 
{\it Examples.} 
a) If $q_*=1$, then \eqref{explicit}  reads
\begin{equation*}
\begin{split}
dv_\bk= \Big(-\lla_\bk v_\bk 
-i\gi \sum_{\bk,\bk',\bk''\in \Z^d_L}
v_{\bk}  v_{\bk'} \bar v_{\bk''}
\delta_{\bk+\bk'\,,\,\bk''+
  r}\,\delta_{\lambda_\bk+\lambda_{\bk'}\,,\,\lambda_{\bk''}+
  \lambda_\bk}\Big)  d\tau
+b_\bk d\bb^\bk\ , 
\end{split}
\end{equation*}
where  $\bk\in \Z^d_L$. 
If $f(t)=t+1$, then this equation  looks similar to the CGL equation 
$$
\dot u-\Delta u+u= i|u|^{2}u +\frac{d}{d\tau}\sum b_\bk \bb^\bk(\tau)e^{i\bk\cdot x}, 
$$
written in the Fourier coefficients. The latter  equation 
 possesses nice analytical properties; e.g.  its stationary 
measures is unique for any $d$, see \cite{KNer13}.

b) Our results remain true if the Hamiltonian $\cH$, corresponding to the nonlinearity in \eqref{1.111},
has variable coefficients. In particular, let $d=1$ and  the  nonlinearity  in   \eqref{1.111}  is replaced by 
$-i p(x) |u|^{2}u$ with a sufficiently smooth function $p(x)$.
Then the effective equation is 
$$
dv_k= \Big(-\lla_k v_k -i \sum_{k_1,k_2,k_3,k_4\in \Z_L}
v_{k_1}  v_{k_2} \bar v_{k_3} p_{k_4}
\delta_{k_1+k_2+k_4\,,\,k_3+k}\,\delta_{k_1^2+k_2^2\,,\,k_3^2+
  k^2}\Big)  d\tau\,
+b_kd\bb^k\ ,
$$
where  $k_L\in \Z_d$ and  $p_k$'s are the Fourier coefficients of $p(x)$.

\section{Main results}\label{s.results}
\subsection{Averaging theorem for  the initial-value problem. 
}\label{s5.2}
We recall that $r$ is a fixed  even integer such that
$r\ge \frac d2+1,$  and abbreviate 
$$
h^r=h,\quad
 C([0,T],h)=  \cH_\ai.
$$
We provide $\cH_\ai$ with the Borel $\sigma$-algebra and the natural
 filtration of the sigma-algebras
$\{\cF_t, 0\le t\le T\}$. 

Let $v^\nu(\tau)$ be a solution of \eqref{5.1} such that
$v^\nu(0)=v_0=\cF(u_0)\in h^{r}$, consider the corresponding
process $a^\nu(\tau)$.
 Due to \eqref{newesti}, the process $a^\nu$ satisfies 
obvious analogies of the estimates \eqref{1.3}, \eqref{2.5} and \eqref{2.05}. 
Since $(R+\cR)(a)$ is  the nonlinearity $P(v)$, written in the $a$-variables,
then 
$$
|(R+\cR)(a)(\tau)|_h = |P(v)(\tau)|_h\le C|v(\tau)|_h^{q_*+1}=  C|a(\tau)|_h^{q_*+1}.
$$
Therefore all moments of $|(R+\cR)(a)|_{\cH_a}$ are finite, and we get from eq.~\eqref{5.21a}
that 
$\ 
\E|a^\nu|_{C^{1/3}([0,T],h)}\le \bar C,
$
uniformly in $\nu$. Now arguing as when proving Lemma~\ref{l7.2} we get that  the 
  set of laws $\cD(\ai^\nu(\cdot)), \  0<\nu\le1$, is tight in
  $\cH_\ai$.  
Consider any limiting measure, corresponding to the laws  $\cD(\ai^\nu(\cdot))$:
\begin{equation} \label{5.88a}
\cD(\ai^{\nu_\ell}(\cdot))\strela \cQ^0_\ai
\as \nu_\ell\to0.
\end{equation}

\begin{theorem}\label{t5.22a}
  There exists a unique weak solution $\ai(\tau)$ of effective
  equation \eqref{5.eff}, satisfying 
 $\ai(0)=   v_0$  a.s. The law of $\ai(\cdot)$ in the space
  $\cH_\ai$ coincides 
  with $\cQ^0_\ai$. The convergence \eqref{5.88a} holds as $\nu\to0$. 
\end{theorem}
The proof of the theorem    is presented  at the end of this section.

Let $\cQ^0$ be a measure in $\cH_{I,V}$ as in
\eqref{5.88}. Since $(I,V)(v^\nu(\cdot))=
(I,V)(\ai^\nu(\cdot))$ for any $\nu>0$ then re-denoting $a(\tau)$ by $v(\tau)$
we derive a corollary from the previous theorem:

\begin{theorem}\label{t5.22}
  There exists a unique weak solution $v(\tau)$ of effective equation \eqref{5.eff}, satisfying 
 $v(0)=   v_0$  a.s. The law of $(I, V )(v(\cdot))$ in the space $\cH_{I,V}$ coincides
  with $\cQ^0$ and the convergence \eqref{5.88} holds as $\nu\to0$. Moreover, for any vectors
  $\tilde s_1,\dots,\tilde s_m
  \in \Z_0^\infty$, perpendicular to $\Lambda$, we have the convergence
$$
\cD(I, V^{\tilde s_1}, \dots, V^{\tilde s_m})
(v^\nu(\cdot))\strela \cD(I, V^{\tilde s_1}, \dots, V^{\tilde s_m})  (v (\cdot)).
$$
\end{theorem}

By this result  the Cauchy problem for the effective equation has a weak solution.
Using 
 Lemma~\ref{l.uniq} and the Yamada-Watanabe argument  (see \cite{KaSh, Yor74, MR99})
we get that  the  equation is well posed:

\begin{corollary}\label{c2}
For any $v_0\in h^r$, eq. \eqref{5.eff} has a unique strong and a unique weak 
solution $v(\tau)$ such that $v(0)=v_0$.  Its law  satisfies \eqref{apriori}. 
\end{corollary}

Now consider $\vp(v^\nu(\tau))\cdot {\tilde s}=\vp(V^{\tilde s}(v^\nu(\tau))\in S^1$. 
Since $\vp(V)$ is a discontinuous function of $V\in \C$, then to pass to
a limit as $\nu\to0$ we do the following. We identity $S^1$ with
$\{v\in \R^2: |v|=1\}$, denote $\lc \tilde s \rc =N$,  and
 approximate the discontinuous function 
 $V^N =(V_1,\dots,V_N)
 \mapsto \vp (V^{\tilde s})$
  by continuous functions 
$$
V^N \mapsto  f_\delta([I(V^N)]) \,
\vp (V^{\tilde s}) 
 \in \R^2\,,
\qquad [I]=\min_{1\le k\le   N } I_k, \quad 0<\delta\ll1. 
$$
where $f_\delta$ is continuous, $0\le f_\delta\le 1$, $f_\delta(t)=0$ 
for $t\le \delta/2$ and $f_\delta=1$ for $t\ge \delta$.

For any measure $\mu_\tau$ in a complete metric space, which weakly continuously depends on $\tau$,
and any $\tau_1<\tau_2$ we will denote
$$
\lan\mu_\tau\ran_{\tau_1}^{\tau_2}=\frac1{\tau_2-\tau_1}
\int_{\tau_1}^{\tau_2} \mu_\tau\, d\tau.
$$
Then the argument above jointly with Lemma~\ref{l5.1} imply:

\begin{corollary}\label{c1}
Let $\ts\in\Z_0^\infty$ be any non-zero vector, orthogonal to $\Lambda$, and let $0\le \tau_1<\tau_2\le T$.
Then 
$$
\lan\,\cD(\vp(u^\nu(\tau))\cdot \ts)\,\ran _{\tau_1}^{\tau_2} \strela 
\lan\,\cD(\vp(v(\tau))\cdot \ts)\,\ran _{\tau_1}^{\tau_2} \quad\text{as}\quad \nu\to0.
$$
\end{corollary}

On the contrary, if $s\cdot\Lambda\ne0$, then by Proposition~\ref{r34} 
we get that 
$$
\lan\,\cD(\vp(u^\nu(\tau))\cdot s)\,\ran _{\tau_1}^{\tau_2}
\strela \dbar \vp.
$$
More generally, if vectors $\ts_1,\dots,\ts_M$ from $\Z_0^\infty$ are perpendicular to $\Lambda$ and
a vector $s$ is not, then
$$
\big\lan\cD(I,\vp\cdot\ts_1,\dots,\vp\cdot \ts_M,\vp\cdot s)(u^\nu(\tau))\big\ran _{\tau_1}^{\tau_2}\strela
\big\lan\cD(I,\vp\cdot\ts_1,\dots,\vp\cdot \ts_M)(v(\tau))\big\ran _{\tau_1}^{\tau_2}\times \, \dbar \vp.
$$
\smallskip

We do not know an equivalent description of the measure $\cQ^0$ only  in terms
of the slow variables $(I,V)$ of equation \eqref{5.1}. But the following result
holds true:

\begin{proposition}\label{p.slow} Consider the natural process on the space $\cH_{I,V}$ with
the measure $\cQ^0$.  If for some  $N\in\N$ and $\delta_*>0$, 
 stopping times $0\le\tau_1<\tau_2\le T$ satisfy
 \eqref{stop}, then for $\tau\in[\tau_1,\tau_2]$ the process $\big(I, \Phi^{(N)}\big)  \big((I,V)(\tau) \big)$ is a
weak solution of the averaged 
equations \eqref{5.21} and \eqref{5.22}${}\mid_{j\le J}$. Here 
$ \Phi^{(N)}=(\Phi_1,\dots, \Phi_{J(N)})$. 
\end{proposition}

 Since    the averaged 
quantities $\langle v_k\cdot P_k\rangle_{\Lambda} $ and  $\langle iv_k\cdot
P_k\rangle_{\Lambda} $ are functions of $I$ and $\Phi$ (see \eqref{AAver}), then equations 
 \eqref{5.21} and \eqref{5.22}${}\!\mid_{j\le J}$  form  an under-determined 
 system of equations for the variables   $(I,\Phi)$.
 \medskip

\noindent
{\it Proof of  Theorem \ref{t5.22a} }
The proof follows the Khasminski scheme (see \cite{Khas68, FW03, KP08}). Its
crucial step is given by  the following lemma:
\begin{lemma}\label{l2.3ihp}
For any $k\ge 1$ one has 
\begin{equation}\label{2.14ihp}
\mathfrak A^\nu_k:=\E \max_{0\le \tau\le  T} \left|\int_0^{\tau }
\cR_k(\ai^\nu(s),\nu^{-1}s) ds  \right| \to 0\quad \mbox{as }
\nu\to 0\ .
\end{equation}
\end{lemma}
The lemma is proved below in Section~\ref{sez:dim}, following the
arguments in \cite{KP08,K12}. Now we derive from it the theorem.

For $\tau\in[0,T]$ consider the processes 
$$
N^{\nu_l}_k=\ai^{\nu_l}_k(\tau) - \int_0^\tau R_k(\ai^{\nu_l}(s))d s\ ,
\quad k \ge 1\ .
$$
Due to \eqref{5.21a} we can write $N^{\nu_l}_k$ as
$$
N^{\nu_l}_k(\tau)= \widetilde N^{\nu_l}_k(\tau)+\overline
N^{\nu_l}_k(\tau)\ ,
$$
where $\widetilde N^{\nu_l}_k(\tau)=a^{\nu_l}(\tau)- \int_0^\tau (
R_k(\ai^{\nu_l}(s))+\cR_k(\ai^{\nu_l}(s),\nu_l^{-1}s)) ds $ is a $\cQ^0_\ai$
martingale and the disparity $\overline N^{\nu_l}_k$ is 
$$ 
\overline
N^{\nu_l}_k(\tau)=\int_0^{\tau }
\cR_k(\ai^{\nu_l}(s))
 ds \ .
$$
The convergence $\cD(\ai^\nu_l)\strela \cQ^0_\ai$ and Lemma~\ref{l2.3ihp}
imply that the processes
$$
N_k(\tau)= \ai_k(\tau)-\int_0^\tau R_k(\ai)ds\ , \quad k\ge 1\ ,
$$
are $\cQ^0_\ai$ martingales (see  for details \cite{KP08}, Proposition 6.3). 

Similar to \eqref{2.14ihp}, we find that
$$
\E \max_{0\le \tau\le  T} \left|\int_0^{\tau }
\cR_k(\ai^\nu(s),\nu^{-1}s) ds  \right|^2 \to 0\quad \mbox{as }
\nu\to 0\ .
$$
Then, using the same arguments as before, we see that the processes
$N_{k_1} (\tau) N_{k_2}(\tau)- \int_0^\tau A_{k_1k_2}ds$ are $\cQ^0_\ai$
martingales, where $A_{k_1k_2}$ denotes the diffusion matrix for the
system \eqref{5.eff}. That is, $\cQ^0_\ai$ is a solution of the martingale
problem with drift $R_k$ and the diffusion $A$. Hence, $\cQ^0_\ai$ is a law
of a weak solution of eq. \eqref{5.eff}. Such a solution exists for any $v_0\in h$.
So by  Lemma~\ref{l.uniq}
and the Yamada -Watanabe argument (see \cite{KaSh,Yor74,MR99}), weak
and strong solutions for \eqref{5.eff} both exist and are unique. Hence, 
 the limit in \eqref{5.88} does not depend on the sequence $\nu_l\to 0$, the 
 convergence holds as $\nu\to0$, 
  and the theorem is proved.
\qed

\subsection{Averaging theorem for stationary solutions.
}\label{s.stat}
Let $v^\nu(\tau)$ be a stationary solution of eq. \eqref{5.1} as at the end of Section~\ref{s1.1}.\footnote{  Under certain 
restrictions on the equation it is  known that its law (i.e. the stationary measure of the equation)
 is unique,
e.g., see \cite{Sh06}. We will not discuss this now.}  Solutions $v^\nu$  inherit the a-priori estimates 
\eqref{1.3}, 
\eqref{2.5}, \eqref{2.05}, so still the set of laws 
$\cD(I(v^\nu(\cdot)), V(v^\nu(\cdot)))$, $0<\nu\le1$, is tight in $\cH_{I,V}$ (cf. Lemma~\ref{l7.2}).  Consider 
any limit
\begin{equation}\label{Conv}
\cD\big(I(v^{\nu_\ell}(\cdot)), V(v^{\nu_\ell}(\cdot)) \big)\strela \cQ\quad \text{as $\nu_\ell\to0$}. 
\end{equation}
As before, the measure $\cQ$ satisfies \eqref{apriori} (with the constants $C_n, C', C^{''}$,  
corresponding  to $v_0=0$).  Moreover, it is stationary in $\tau$.

\begin{theorem}\label{t.stat}
 There exists a stationary solution $v(\tau)$ of the effective equation \eqref{5.eff} such that 
$\cQ=\cD\big(I(v(\cdot)), V(v (\cdot)) \big)$. 
\end{theorem}
\begin{proof}
Denote $\mu^\nu=\cD v^\nu(\tau)$. Estimate \eqref{2.5} with $2m=r$ and 
 $n=1$ implies that 
$
\int |v|^2_{h^{r+1}}\,\mu^\nu(dv)\le C
$
for all $\nu$. So the set of measures $\mu^\nu$ is tight in $\cH^r$. Replacing, if necessary, 
the sequence $\{\nu_l\}$ by a subsequence, we achieve that 
\begin{equation}\label{ko-ko}
 \mu^{\nu_l}\strela \mu^0\quad\text{as}\quad \nu_l\to0. 
\end{equation}
Clearly $(I,V)\circ \mu^0$ is the marginal distribution for $\cQ$ as $\tau=\,$const, which we will denote $q$
(i.e., $q=\cQ\mid_{\tau=\const}$). 

Let $v^0(\tau), \tau\ge0$, be a solution for the effective equation \eqref{5.eff} such that $\cD v^0(0)=\mu^0$
(existing by Corollary~\ref{c2} and the estimates on $\mu^0$).  Then, for the same reason as in Section~\ref{s5.2}, 
$$
\cD\big(I,V)(v^0(\tau) )\mid_{\tau\in[0,T]}=\cQ,
$$
and $\cD(I,V)(v^0(\tau))\equiv q$. We do not know if the solution $v^0$ is stationary, but from the Bogolyubov-Krylov
argument we know that for a suitable sequence $T_j\to\infty$ we have the convergence
$$
\frac1{T_j} \int_0^{T_j} \cD(v^0(\tau))\,d\tau \strela m^0,
$$
where $m^0$ is a stationary measure for \eqref{5.eff}. Still we have that $(I,V)\circ m^0=q$, and the measure $m^0$
satisfies the same apriori estimates as before. Let $v(\tau)$ be a solution for  \eqref{5.eff} such that $\cD v(0)=m^0$.
It is stationary and $\cD(I,V)(v(\tau))\equiv q$.
 Modifying a bit the argument above we get that also
$\cD(I,V)(v(\cdot))=\cQ$.
\end{proof}

	Writing the convergence \eqref{Conv} as
	$
	\cD(I,V)\big(v^{\nu_l} (\cdot)\big) \strela \cD(I,V)\big(v(\cdot)\big),
	$
we note that, as in Section~\ref{s5.2}, we also have that 
$$
\cD(I,V^{\tilde s_1},\dots, V^{\tilde s_m}   )\big(v^{\nu_l} (\tau)\big) \strela
\cD (I,V^{\tilde s_1},\dots, V^{\tilde s_m} ) \big(v (\tau)\big) = (I,V^{\tilde s_1},\dots, V^{\tilde s_m} ) \circ m^0
$$
as $\nu_l\to0$, 
for any $m$ and any vectors $\tilde s_1,\dots,\tilde s_m$, perpendicular to $\Lambda$. Since for stationary solutions
$v^\nu(\tau)$ we have $\langle\cD(v^\nu(\tau))\rangle_{\tau_1}^{\tau_2} = \cD(v^\nu(\tau))$, then arguing as when 
proving Corollary~\ref{c1} we also get that 
\begin{equation}\label{zzz}
\cD (I,\Phi^{\tilde s_1},\dots, \Phi^{\tilde s_m} ) \big(v^{\nu_l} (\tau)\big) \strela
 (I,\Phi^{\tilde s_1},\dots, \Phi^{\tilde s_m} ) 
 \circ m^0\,.
\end{equation}
Moreover if $s\in\Z_0^\infty$ is such that $s\cdot\Lambda\ne0$, then in view of Proposition~\ref{r34} and the
stationarity of the solutions we have
\begin{equation}\label{zz}
\cD (I,\Phi^{\tilde s_1},\dots, \Phi^{\tilde s_m}, \Phi^s ) \big(v^{\nu_l} (\tau)\big) \strela\big(
 (I,\Phi^{\tilde s_1},\dots, \Phi^{\tilde s_m} )\circ m^0\big)  \times  \,\dbar \theta\,.
\end{equation}

If eq. \eqref{5.eff} has a unique stationary measure $m^0$, then the convergences above hold as $\nu\to0$.
But in this case a stronger assertion holds:

\begin{theorem}\label{t.univ}
Let $v^\nu$ be a stationary solution of equation 
 \eqref{5.1}, $\cD (v^\nu(\tau))\equiv \mu^\nu$, 
and assume that the effective equation 
\eqref{5.eff} has a unique stationary measure $m^0$. Then 
 \begin{equation}\label{lim}
 \mu^\nu \strela   m^0 \qquad \text{as} \quad \nu\to0.
\end{equation}
\end{theorem}

\begin{proof} i)
Consider again the convergence \eqref{ko-ko}. We are going to show that  the limiting measure $\mu^0$ 
 equals $m^0$. Then the limit in \eqref{ko-ko} does not depend on
the sequence $\{\nu_l\to0\}$, so it holds  as $\nu\to0$, and \eqref{lim} follows.

ii)   
Due to Lemma \ref{l5.1}, $\mu^\nu(\Game)=0=\mu^0(\Game)$, so we may regard $\mu^\nu$
and $\mu^0$ as measures on $h_I^r\times \T^\infty$.  
Let us fix any $n\in\N$ and consider measures $\mu^{\nu\, n}$,
 $\mu^{0\,n}$ and $m^{0\,n}$ which are images of the measures $\mu^\nu$, $\mu^0$ and $m^0$ 
 under the projection 
 $$
 \Pi^n:v\mapsto v^n.
 $$
 We will regard them  as measures on $\R^n_+\times \T^n=\{(I^n,\vp^n)\}$.  To prove that $\mu^0=m^0$ 
 it suffices to verify that   $\mu^{0\,n}=m^{0\,n}$ for each $n$. 
 
 Let us  denote
$\cA(\Lambda^n)=:\cA^n$, and 
let  the vectors $\zeta^1,\dots,\zeta^n\in\Z^n$ and the unimodular 
 matrix $R$ be as in Lemma~\ref{l1.1} with $\cA=\cA^n$. Let
 $L=L_{\cA^n}:\T^n\to \T^{n-1}$ be the operator in \eqref{f0}, i.e.
\be\label{L}
L:
\T^n\ni \vp^n\mapsto  (\vp^n\cdot \zeta^1,\dots, \vp^n\cdot \zeta^{n-1})^T\in\T^{n-1}\,.
\ee
Writing $R^T(\vp^n)=(y_1,\dots,y_n)^T=( {\mathbf y}, y_{n})^T$, where ${\mathbf y}=(y_1,\dots,y_{n-1})^T$, 
we have $L(\vp^n)={\mathbf y}$. We will denote  by $\pi^1$ the natural projection $y\mapsto {\mathbf y}$.

 For further purposes we make the following observation. Let $\mu$ be a Borel measure on $h$. Consider its images under rotations $\Psi_{t\Lambda}$ and projections $\Pi^n$. In the $(I,\vp)$-variables the mapping
 $\Psi_{t\Lambda}$ becomes $\,$id$\,\times (\cdot+t\Lambda)$, so
 $$
 \Pi^n\circ (\Psi_{t\Lambda}\circ \mu) =\big( \text{id}\, \times  (\cdot+t\Lambda^n) \big)\circ\Pi^n \circ\mu
 $$
(where $\Pi^n\circ \mu$ is written in the $(I^n, \vp^n)$-variables). By \eqref{resolve} the transformation $R^T$ 
of $\T^n$ conjugates the translation by the vector $t\Lambda^n$ with the translation by $t \,e^n$. Therefore,
\be\label{conjug}
\cR^T\circ \Pi^n\circ (\Psi_{t\Lambda} \circ \mu)= (\text{id}\, \times (\cdot+t\, e^n)) \circ\cR^T\circ\Pi^n \circ\mu\,,
\ee
where $\cR=\,$id$\,\times R^T$. 
 
 iii) 
 Let us  apply to the measures   $\mu^{\nu\, n}$,  $\mu^{0\,n}$, $m^{0\,n}$  the 
 transformation $\cR^T$:
\be\label{xxi}
N^{\nu\,n}={\cR^T}\circ \mu^{\nu\,n}\,,
\quad N^{0\,n}={\cR^T}\circ \mu^{0\,n}\,,
\quad M^{0\,n}={\cR^T}\circ m^{0\,n}\,.
\ee
Recall that by \eqref{ko-ko}, 
$\ 
N^{\nu_l\, n}\strela N^{0\,n}$  as  $\nu_l\to0\,.$
Our first goal is to calculate the limiting measure $N^{0\,n}$. To do this  let us 
 disintegrate   $N^{\nu\,n}$ and $ N^{0\,n}$
 with respect to the mapping 
$$
\text{id}\,\times \pi^1: \R^n_+\times \T^n \to \R^n_+\times \T^{n-1},\quad
(I^n,({\mathbf y}, y_{n})^T) \mapsto (I^n, {\mathbf y}).
$$
That is (see \cite{Dud}, Section 10.2), write them as
$$
N^{\nu\,n} =N^{\nu\,n}_{I^n,{\mathbf y}}(dy_{n})\,p^{\nu\,n}(dI^n\,d{\mathbf y}),\quad
N^{0\,n} =N^{0\,n}_{I^n,{\mathbf y}}(dy_{n})\, p^{0\,n}(dI^n\,d{\mathbf y})\,,
$$
where $p^{\nu\,n} = (\text{id}\,\times \pi^1)\circ N^{\nu\,n}$ and 
$p^{0\,n} = (\text{id}\,\times \pi^1)\circ N^{0\,n}$
Since ${\mathbf y}=L(\vp^n)$, then $p^{\nu\,n}=\cD(I^n\times (L\circ\vp^n))(v^{\nu\, n} (\tau))$. As each vector $\zeta^j$ in \eqref{L}
is perpendicular to $\Lambda^n$, then in view of \eqref{zzz} we have
\begin{equation}\label{x1}
p^{0\,n}=\lim_{\nu_l\to0} \cD(I^n \times (L\circ\vp^n) )(v^{\nu_l \,n}(\tau)) = (I^n\times (L\circ\vp^n))\circ m^{0\,n}. 
\end{equation}

To calculate  $N^{0\,n}$ it
 remains to find the fiber-measures $N^{0\,n}_{I^n, {\mathbf y}}$. To do this let us take any bounded continuous function $f$ 
 on 
$\R^n_+\times \T^{n-1}\times S^1$ and consider 
$\ 
\langle N^{\nu\,n}, f\rangle = \E f(I^n, {\mathbf y}, y_{n})(v^\nu(\tau)).
$
Since ${\mathbf y}(v)=L(\vp^n)$ and 
$y_{n}(v)= v\cdot\eta^n$, where the vector $\eta^n$  is not perpendicular to $\Lambda$, then 
by \eqref{zz}
$$
\langle  N^{\nu\,n} , f\rangle  \to \int f(I^n,{\mathbf y}, y_{n})
\Big(\big(I^n\times (L\circ\vp^n)\big)   m^{0\,n} \Big)(dI^n\,d{\mathbf y})\,
 \dbar y_{n}. 
$$
From other hand, by \eqref{ko-ko} 
$$
\langle   N^{\nu_l\,n},  f\rangle  \to \langle N^{0\,n}, f\rangle = \int f(I^n, {\mathbf y}, y_{n})
 \, N^{0\,n}_{I^n, {\mathbf y}}(dy_{n})\, p^{0\,n}(dI^n\,d{\mathbf y})\,.
$$
Since $p^{0\,n}=(I^n\times (L\circ\vp^n))\circ m^{0\,n}$, then
we get from the two convergences above that 
 for $p^{0\,n}$-a.a. pairs $(I^n,{\mathbf y})$ we have  $N^{0\,n}_{I^n,{\mathbf y}}=\dbar y_{n}$. 
 Accordingly, 
$$
N^{0\,n} = \dbar y_{n} \times   p^{0\,n}(dI^n\,d{\mathbf y}) \,.
$$

iv)
Consider the measure $M^{0\,n}$. Due to \eqref{x1} its disintegration with respect to the mapping 
id$\,\times \pi^1$ may be written as 
\begin{equation}\label{??}
M^{0\,n} =M^{0\,n}_{I^n,{\mathbf y}}(dy_n)p^{0\,n}(dI^n\,d{\mathbf y})
\end{equation}
with some unknown fiber-measures $M^{0\,n}_{I^n,{\mathbf y}}$. Now consider the rotated 
 measure $\Psi_{t\Lambda} \circ m^{0} $, $t\ge0$,  and its $n$-dimensional projection.
By \eqref{conjug},
$$
\cR^T\circ\Pi^n\circ \Psi_{t\Lambda}\circ m^{0}= (\text{id}\, \times l_t)\circ\cR^T\circ m^{0\,n}, 
$$
where $l_t ({\mathbf y},y_n) = ({\mathbf y},y_n+t)$. Due to \eqref{xxi} and 
\eqref{??}, the measure in the r.h.s. 
equals 
$$
M^{0\,n}_{I^n,{\mathbf y}}(dy_n +t)p^{0\,n}(dI^n\,d{\mathbf y})\,.
$$
But by Corollary~\ref{c.invar}, the measure in the l.h.s. does not depend on $t$. 
So
 $ M^{0\,n}_{I^n,{\mathbf y}}(dy_n)\equiv M^{0\,n}_{I^n,{\mathbf y}}(dy_n+t)
 $ is a translation-invariant
measure on $S^1$, and it must be equal to $\dbar y_n$.  Accordingly, 
$$M^{0\,n} = \dbar y_n\times p^{0\,n}(dI^n\, dy^n)=
N^{0\,n}.
$$

v)
We have established that  $N^{\nu_l\,n}\strela M^{0\,n}$ as $\nu_l\to0$. So $\nu^{\nu_l\,n}\strela m^{0\,n}$,
which completes the proof.
\end{proof}

\subsection{Mixing in the effective equations
}\label{s.mix}
We start  with the case when the function $f(\lambda)$ has a linear growth. 
For simplicity of notation we suppose  that  $f(\lambda)=\lambda+1$. We also are forced to assume
that $q_*=1$.

 The effective equation \eqref{5.eff}=\eqref{*eff1} with
$q_*=1$ looks similar to the equation \eqref{1.100}${}_{\nu=\infty, q_*=1}$, studied in \cite{KNer13}. It turns out that 
the two equations indeed are similar, at least for $d\le3$, and that the proof of the mixing in Section~4 of \cite{KNer13},
based on an abstract theorem from \cite{KS}, applies to \eqref{5.eff} with minimal changes. Indeed, the 
crucial step 
in \cite{KNer13} in order to apply the result from \cite{KS} is to establish for solutions of the equation the
exponential estimate of the form
\begin{equation}\label{m1}
\PP\{\sup_{t\ge0}(\int_0^t |u(s)|^2_{L_\infty}ds-Kt)\ge\sigma\}\le C' \exp(c_1|u_0|^2_{L_\infty} -c_2\sigma),
\qquad\forall\,\sigma>0,
\end{equation}
with suitable constants $K, C', c_1$ and $c_2$. This estimate is important to study the mixing 
since it allows to control divergence of trajectories $u_1(t)$ and $u_2(t)$, corresponding to 
the same realisation of the random force, through the inequality\footnote{To match \eqref{m1} 
and \eqref{m2} we use crucially that $q_*\le1$.}
\begin{equation}\label{m2}
|u_1(t)-u_2(t)|_{L_2} \le |u_1(0)-u_2(0)|_{L_2}
\exp\big(C\int_0^t (|u_1(s)|^2_{L_\infty} +  (|u_2(s)|^2_{L_\infty})ds\big). 
\end{equation}

For eq. \eqref{5.eff} an analogy of \eqref{m1} follows by applying the Ito formula to $[v]^2_1=H_0(v)+H_1(v)$
(see \eqref{integrals}), since due to \eqref{integrals} we have that 
$$
d[v(\tau)]^2_1 + 2\int_0^\tau [v(s)]_2^2\,ds  =4\tau\cB+2\sum_{j=1}^\infty (\lambda_j+1)(v_j(\tau)\cdot d\bb^j(\tau),
$$
where we denote 
$[v]_2^2 = \sum (\lambda_j+1)^2|v_j|^2$ and $\cB=\sum(\lambda_j+1)b_j^2$.  Applying to this relation the 
supermartingale inequality in the standard way (e.g., see in \cite{KNer13, KS}), we get that 
\begin{equation*} 
\PP\{\sup_{\tau\ge0}(\int_0^\tau [v(s)]^2_2\,ds-2\cB t)\ge\sigma\}\le C' \exp(c_1|v_0|^2_1 -c_2\sigma),
\qquad\forall\,\sigma>0.
\end{equation*}
If $d\le3$,  then by Lemma \ref{estil2}  the 
 divergence of two solutions for \eqref{5.eff} with the same $\omega$ satisfies 
  \begin{equation*} 
  |v_1(\tau)-v_2(\tau)|_{h^0}\le |v_1(0)-v_2(0)|_{h^0}
  \exp\big(C \int_0^\tau ([v_1(s)]_2^2 + [v_2(s)]^2_2) \,ds\big).
  \end{equation*}
  This last two estimates  allow to repeat  literally  for equation 
  \eqref{5.eff} the reduction to Theorem~3.1.3 from \cite{KS},
  made in \cite{KNer13}, and prove
  
  \begin{theorem}\label{t.mix}
  Let $q_*=1$,  $f(\lambda)=\lambda+1$
   and $d\le3$. Then the effective equation \eqref{5.eff} has a unique stationary measure $\mu$ and is mixing. That is,
  every its solution $v(\tau)$ satisfies $\cD(v(\tau))\strela \mu$ as $\tau\to \infty$. 
  \end{theorem}

  The presented proof uses  that the nonlinearity in the effective equation is at most cubic. It 
   also applies  to the effective equations for eq.~\eqref{1.100}, where the Hamiltonian
   $\cH$ is one of the two functions $\cH^3$ with cubic densities as at the end of Section~\ref{s0.1}
   (in this case the argument works if $d\le6$). The proof without changes applies 
   to equation \eqref{1.100}, where $q_*=1$, $d\le3$ and $f(\lambda)$ grows super-linearly. The
   argument also may be adjusted to the case when $q_*=1$, $d$ is any and $f(\lambda)= c_1+ \lambda^{c_d}$,
   where $c_d$ is sufficiently big. Based on the similarity with the equation  \eqref{1.100}${}_{\nu=\infty, q_*=1}$, studied 
   in \cite{KNer13} in for any space-dimension, we conjecture that for  $q_*=1$ and  $f(\lambda)=\lambda+1$ the 
   effective equation is well-posed and mixing for any $d$.
    But it is unknown how to prove the mixing for equations with 
   $q_*\ge2$ (in any space-dimension).

\subsection{Proof of Lemma~\ref{l2.3ihp}}\label{sez:dim}
For this proof we adopt a notation from \cite{KP08}. Namely, 
we denote by $\vk(t)$ various functions of $t$ such that $\vk\to0$ 
as $t\to\infty$, and denote by $\vk_\infty(t)$ functions, satisfying  
$\vk(t)=o(t^{-N})$ for each $N$. We write $\vk(t,M)$ to indicate that $\vk(t)$ 
depends on a parameter $M$. Besides for events $Q$ and $O$ and  a
random variable  $f$ we write $\PP_O(Q)=\PP(O\cap Q)$ and 
$\E_O(f)=\E(\chi_O\, f)$. Below $M$ stands for a suitable function of
$\nu$ such that  $M(\nu)\to \infty$ as $\nu\to 0$, but
$$
\nu M^n\to0 \quad \mbox{as } \nu\to 0\ , \quad \forall n\ .
$$

Denote by $\Omega_M=\Omega^\nu_M$ the event
$$
\Omega_M=\left\{\sup_{0\le\tau\le T} |\ai^\nu(\tau)|_{h^r}\le M\right\}\ .
$$
Then, by \eqref{2.05}, $\PP(\Omega^c_M)\le \vk_\infty(M)$ uniformly in
$\nu$, so that one has
$\ 
\mathfrak A^\nu_k\le \vk_\infty(M)+ \mathfrak
A^\nu_{k,M}\ ,
$
where we have defined
\begin{equation}\label{5.169}
\mathfrak A^\nu_{k,M}:= \E_{\Omega_M} \max_{0\le \tau\le T} \left|\int_0^{\tau}
 \cR_k(\ai^{\nu}(s),\nu^{-1}s)ds
\right|\,.
\end{equation}
So it remains to estimate $A^\nu_{k,M}$. 

Consider a partition of $[0,T]$  by the points
$$
 \tau_n= nL,\quad 0\le n\le K\sim T/ L. 
$$
where $ \tau_{K}$ is the last point $ \tau_n$ in $[0,T)$.
  The diameter $L$ of the partition is 
$\ 
L=\sqrt\nu. 
$
  Denoting
\begin{equation}\label{5.17}
\eta_l= \int_{\tau_l}^{ \tau_{l+1}}  \cR_k(\ai^{\nu
  }(s),\nu^{-1}s)ds\ ,\quad 
0\le l\le K-1\ ,
\end{equation}
we see that 
\begin{equation}\label{5.170}
\mathfrak A^\nu_{k,M} \le LC(M)+\E_{\Omega_M}\sum_{l=0}^{K-1}|\eta_l|\ ,
\end{equation}
since for $\omega\in \Omega_M$ the integrand in
\eqref{5.17} is smaller than a suitable $C(M)$ (see Lemma~\ref{l.P^0} and
\eqref{p4}).  For any $l$  let us consider the event
$$
\cF_l=\{\sup_{\tau_l\le\tau\le\tau_{l+1}}|\ai^\nu(\tau)-\ai^\nu(\tau_l)|_{h}\ge
P_1(M)L^{1/3}\}\ ,
$$
where $P_1(M)$ is a suitable polynomial.  It is not hard to verify
using the Doob inequality that for a suitable choice of  $P_1$  the 
probability of
 $\IP(\cF_l) $  is less than   $ \vk_\infty(L^{-1};M)$ 
(cf. \cite{KP08}).  One gets
\begin{equation}\label{5.211}
\sum_{l=0}^{K-1}\left|\E_{\Omega_M}|\eta_l|-
\E_{\Omega_M\backslash \cF_l}|\eta_l|  \right| \le
{C(M)}{L}\sum_{l=0}^{K-1}\IP(\cF_l)\le C(M) \vk_\infty(L^{-1};M)\ , 
\end{equation}
so that it remains to estimate $\sum\E_{\Omega_M\backslash \cF_l} |\eta_l|$. 

We have
\begin{equation*}
 \begin{split}
\  |\eta_l|   &\le \left|
\int_{\tau_l}^{\tau_{l+1}} \left(  \cR_k(\ai^{\nu
  }(s),\nu^{-1}s)- \cR_k(\ai^{\nu
  }(\tau_l),\nu^{-1}s)\right) ds\right|\\
&+ \left|
\int_{\tau_l}^{\tau_{l+1}}\left(  \cR_k(\ai^{\nu
  }(\tau_l),\nu^{-1}s)\right) ds\right| 
 =:\Upsilon^1_l+\Upsilon^2_l\ .
\end{split}
\end{equation*}
By the regularity of the integrand and the definition of $\cF_l$
\begin{equation}\label{5.002}
  \sum_l  
\E_{\Omega_M \backslash\cF_l}  
 \Upsilon^1_l  \le \vk(L^{-1/3};M)=\vk( \nu^{-1/6};M) \ . 
\end{equation}
So it remains to estimate the expectation of $\sum\Upsilon^2_l$. 
  Denoting $t=\nu\tau$ and making use of \eqref{eq:nonres} we  write
  $\Upsilon^2_l$ as
\begin{equation*}
\begin{split}
\Upsilon^2_l&=L\Biggl|\frac{\nu}{L}\int_{0}^{\nu^{-1}L}\sum_{\substack{p,q,l\in\ZZ \\
     q-l -e^k\not \in\cA(\Lambda,m)\\ |q|+|l|+1\le m}}   
P^{0pql}_k(\ai)\exp\Bigl(-it\left( 
\Lambda \cdot (q- l- e^k)\right)\Bigr)dt\Biggr|\\
&\le L C(M) \frac{\nu}{L} \sup_{\substack{p,q,l\in\ZZ \\
     q-l -e^k\not \in\cA(\Lambda,m)\\ |q|+|l|+1\le m}}
\frac{1}{\Lambda \cdot (q- l- e^k)} \le
L\vk(\nu^{-1}L;M)\ ,
\end{split}
\end{equation*}
because the supremum in the second line is bounded by one, since both
$\Lambda$ and $q-l-e^k$ are integer vectors.
Therefore
\begin{equation}\label{5.003}
 \sum_l  \E_{\Omega_M\backslash \cF_l}
\Upsilon^2_l  \le\vk(\nu^{-1/2};M). 
\end{equation}

Now \eqref{5.169}, \eqref{5.170},  \eqref{5.211}, \eqref{5.002} and
\eqref{5.003} imply that 
$$
\mathfrak A^\nu_k\le \vk_\infty(M)+\vk(\nu^{-1/2};M) +
\vk_\infty(\nu^{-1};M) +\vk( \nu^{-1/6};M) +\vk(\nu^{-1/2};M)\ . 
$$
Choosing first $M$ large and then $\nu$ small,  
 we make the r.h.s. above  arbitrarily small. This proves the lemma.
\qed

An argument similar to the previous one (see Proposition~4.7 of
\cite{KM13}) implies the following assertion:  
\begin{proposition}\label{r34}
Let $s\in\Z_0^\infty$ be such that $s\cdot\Lambda\ne0$ and 
$G:\R_+^M\times \T^{J(M)}\times S^1\to\R$ be a bounded Lipschitz-continuous function,
for some $M\ge1$. Then
\begin{equation*}
\begin{split}
\mathfrak B^\nu:=
\E\max_{0\le\tau\le T} &\Big|
 \int_0^\tau \Big( G(I^{\nu M}(l), \Phi^{\nu (M)}(l),s\cdot \vp^\nu(l))-\\
 &\int_{S^1}  G(I^{\nu M}(l), \Phi^{\nu (M)}(l), \theta)\,\dbar\theta\Big)dl\Big|\to0\quad\text{as}\quad
 \nu\to0. 
 \end{split}
\end{equation*}
In particular, taking for $G$  Lipschitz functions on $S^1$ we get that 
$
\lan
\cD(s\cdot\vp^\nu(l)) \ran_0^t
\strela d\theta$ as  $\nu\to0, 
$
for any $t>0$. 
\end{proposition}

\bibliography{meas}
\bibliographystyle{amsalpha}
\end{document}